\documentclass{elsarticle}

\usepackage{amsmath, amssymb}

\usepackage{paralist}
\usepackage{graphicx}
\usepackage{epsfig}

\usepackage{framed}
\usepackage{tabularx}

\newcommand{\problemDecision}[3]
{\noindent
 \begin{framed}
   \centerline{{\sc #1}}     
   \bigskip
   \noindent
   \begin{tabularx}{\textwidth}{lX}
     \emph{Given:} & #2	\\
     \emph{Question:} & #3
   \end{tabularx}
 \end{framed}
}

\newtheorem{case}{Case}
\newtheorem{subcase}{\textbf{Case}}
\numberwithin{subcase}{case}
\newtheorem{subsubcase}{\textbf{Case}}
\numberwithin{subsubcase}{subcase}
\newtheorem{definition}{Definition}
\newtheorem{theorem}{Theorem}
\newtheorem{lemma}{Lemma}
\newtheorem{corollary}{Corollary}
\newproof{proof}{Proof}

\newcommand{\ep}[2]{\overline{(#1,#2)}}

\begin{document}

\begin{frontmatter}
\title{{\sc Metric Dimension} for Gabriel Unit Disk Graphs is NP-Complete}

\author{Stefan Hoffmann\corref{cor1}}
\ead{stefan.hoffmann@hhu.de}
\author{Egon Wanke}
\ead{E.Wanke@hhu.de}

\cortext[cor1]{Corresponding author, phone number +492118110684}
\address{Institute of Computer Science, Heinrich-Heine-Universit\"at D\"usseldorf, D-40225 D\"usseldorf, Germany}

\begin{abstract}

We show that finding a minimal number of landmark nodes for a unique virtual addressing by hop-distances in wireless ad-hoc sensor networks is NP-complete even if the networks are unit disk graphs that contain only Gabriel edges. This problem is equivalent to {\sc Metric Dimension} for Gabriel unit disk graphs.
The Gabriel edges of a unit disc graph induce a planar $O(\sqrt{n})$ distance and an optimal energy spanner. This is one of the most interesting restrictions of {\sc Metric Dimension} in the context of wireless multi-hop networks.
\end{abstract}

\begin{keyword}
metric dimension \sep unit disk graph \sep gabriel graph \sep virtual address \sep wireless multi-hop network \sep sensor network
\end{keyword}
\end{frontmatter}

%\maketitle
%%%%%%%%%%%%%%%%%%%%%%%%%%%%%%%%%%%%%%%%%%%%%%%%%%%%%%%%%%%%%%%%%%%%%%%%%%%%%%%%

\section{Introduction}

Wireless radio networks in which all nodes have the same radio range are often modeled as {\em unit disc graphs}. An undirected graph is a unit disk graph (UDG) if its vertices can be embedded in the Euclidean plane $\mathbb R^2$ by an embedding $\rho$ such that two vertices $u,v$ are connected by an edge if and only if their Euclidean distance $d^\text{e}(u,v):=\|\rho(u)-\rho(v)\|_2$ is at most $1$. The vertices represent the sensor nodes and the undirected edges the symmetric communication channels between them. Two sensor nodes can communicate with each other, if they are close enough.

A widely used idea to reduce the complexity of distributed algorithms for wireless ad-hoc sensor networks is to consider only some of the available connections. The strongest restriction that preserves connectivity is a spanning tree. A spanning tree unfortunately does not allow efficient routing through the network. A slightly weaker restriction is the {\em Gabriel graph} \cite{GS69}. An edge $\{u,v\}$ of an embedded graph $G$ is a {\em Gabriel edge}, if there is no vertex $w$ such that $d^e(u,w)^2+d^e(w,v)^2 \quad \leq \quad d^e(u,v)^2$. The Gabriel edges induce a planar $O(\sqrt{n})$ distance and optimal energy spanner $G'$ of $G$, see \cite{BDEK06}. That is, the ratio between the length of a shortest path in $G'$ and $G$ is of $O(\sqrt{n})$ if the costs of the edges $\{u,v\}$ are $d^e(u,v)$, and $1$ if the costs of the edges are $d^e(u,v)^\alpha$ for some $\alpha \geq 2$.

Many routing algorithms for wireless ad-hoc sensor networks are based on {\em virtual coordinates}, see for example VCap \cite{CCDU05}, JUMPS \cite{BPDCFS06}, GLIDER \cite{FGGSZ05}, VCost \cite{EMS07}, BVR \cite{FRZ05}, and HBR \cite{GHW12}. The virtual coordinates are computed from the distances to specific nodes called {\em landmarks}, {\em anchors}, or {\em beacons} \cite{FRZ05,KRR96}. The setup of the virtual addressing starts by selecting some landmark nodes that flood large parts of the network. After that every node defines its virtual coordinates depending on the distances to the landmark nodes.

The virtual coordinates can be used to route a message through the network. Most routing algorithms require the virtual addresses to be unique. Packet delivery is also only guaranteed if different nodes have different virtual coordinates. The length of the virtual addresses increases with the number of landmark nodes. Since every landmark node has to flood large parts of the network during initialization, the number of landmark nodes should be as small as possible to reduce the amount of consumed energy. These conditions cause the question of how to determine a minimum set of landmark nodes that provides a unique virtual addressing.

From a graph theoretical point of view, the decision version of the problem above is called {\sc Metric Dimension}. A set of vertices $S \subseteq V$ of a graph $G=(V,E)$ is called a {\em resolving set}, if for every vertex pair $u,v \in V$ there is at least one vertex $s\in S$ such that the distance between $s$ and $u$ differs from the distance between $s$ and $v$. A graph has metric dimension at most $k$ if there is a resolving set of size at most $k$. {\sc Metric Dimension} is known to be NP-complete for general graphs as well as for planar graphs \cite{DPSL12,GJ79,KRR96}. For a fixed number of landmarks $k$ it is decidable in polynomial time. It is also decidable in polynomial time for special classes of graphs like trees, wheels, complete graphs, $k$-regular bipartite graphs and outerplanar graphs \cite{CEJO00,DPSL12,KRR96,SBSSB11,SSH02}. The approximability of {\sc Metric Dimension} by centralized algorithms has been studied for bounded degree, dense, and general graphs in \cite{HSV11}. There are also 
boundaries for the metric dimension of some special classes of graphs, see \cite{CGH08,CPZ00}.

In this paper, we prove that {\sc Metric Dimension} is NP-complete even for unit disk graphs that contain only Gabriel edges. This is the most interesting restriction of the problem in the context of computing unique virtual addresses in wireless ad-hoc sensor networks. We show how to construct in polynomial time for an arbitrary instance $(X,{\cal C})$ of a special satisfiability problem a Gabriel unit disk graph $H_\psi$ and an integer $k$ such that there is a satisfying truth assignment for $X$ if and only if there is a resolving set $S$ for $H_\psi$ of size at most $k$.

%%%%%%%%%%%%%%%%%%%%%%%%%%%%%%%%%%%%%%%%%%%%%%%%%%%%%%%%%%%%%%%%%%%%%%%%%%%%%%%%

\section{Definitions and terminology}

Let $G=(V,E)$ be a {\em directed} or {\em undirected graph}. That is, $V$ is the set of vertices and $E$ is the set of {\em directed edges} $E\subseteq \{(u,v)~|~u,v \in V,~u \not=v\}$  or {\em undirected edges} $E \subseteq \{\{u,v\}~|~u,v \in V,~u \not=v\}$, respectively. A graph $G'=(V',E')$ is a {\em subgraph} of $G=(V,E)$ if $V' \subseteq V$ and $E' \subseteq E$. It is an {\em induced subgraph} of $G$ if $E' = E \cap \{(u,v)~|~ u,v \in V'\}$ or $E' = E \cap \{\{u,v\}~|~ u,v \in V'\}$, respectively.

A sequence of $k$ vertices $(v_1, \dots ,v_k)$ of $G$ is called a {\em directed path} if $G$ is directed and $(v_i,v_{i+1})\in E$ for $1\le i < k$. It is called an {\em undirected path} if $G$ is undirected and $\{v_i,v_{i+1}\}\in E$ for $1\le i < k$.

\subsection{Gabriel unit disk graphs}
%For a directed graph $G=(V,E)$, a sequence of $k$ vertices $(v_1, \dots ,v_k)$  is called a {\em chain} if $(v_i,v_{i+1})\in E$ or $(v_{i+1},v_i)\in E$ for $1\le i < k$.
For a vector $(x,y) \in  \mathbb{R}^2$, let $\Vert (x,y)\Vert_2 := \sqrt{x^2 + y^2}$ be its {\em Euclidean norm}.
\begin{definition}\label{DefinitionUDG}
An undirected graph $G=(V,E)$ is a {\em unit disk graph (UDG)} if there is a {\em UDG embedding} $\rho: V \to \mathbb{R}^2$ such that
$\forall u,v \in V$, $(u \not = v) \Rightarrow (\rho(u) \not= \rho(v))$ and
$\{u,v\}\in E \Leftrightarrow \Vert \rho(u) - \rho(v) \Vert_2 \le 1$.
%The embedding $\rho$ of a UDG is called a {\em UDG embedding}.
\end{definition}

%If $G$ is a planar UDG, then obviously every induced subgraph of $G$ is a planar UDG.

\begin{definition}\label{GabrielGraph}
For an undirected graph $G=(V,E)$ and an embedding $\rho: V \to \mathbb{R}^2$ of its vertices an edge $\{u,v\}$ is a {\em Gabriel edge} if there is no vertex $w \in V$ with
\[\Vert \rho(u) - \rho(w) \Vert_2^2 + \Vert \rho(w) - \rho(v) \Vert_2^2 \quad \leq \quad \Vert \rho(u) - \rho(v) \Vert_2^2.\]

An undirected graph $G=(V,E)$ is a {\em Gabriel UDG (GUDG)} if there is a UDG embedding $\rho: V \to \mathbb{R}^2$ for which all edges in $E$ are Gabriel edges. Such an embedding is called a {\em GUDG embedding}.
\end{definition}

\begin{definition}\label{DefinitionPlanarity}
A graph $G=(V,E)$ is {\em planar} if there is an embedding $\rho: V \to \mathbb{R}^2$ with the following two properties:
\begin{enumerate}
\item $\forall u,v \in V$, $(u \not = v) \Rightarrow (\rho(u) \not= \rho(v))$ and 
\item
for every edge between two vertices $u$ and $v$ %$\{u,v\} \in E$ or $(u,v) \in E$ 
there is a line in $\mathbb{R}^2$ connecting the points $\rho(u)$ and $\rho(v)$ such that no two lines for two different edges intersect except at a common endpoint.
\end{enumerate}
%For two points $p,q\in \mathbb{R}^2$ let $\overline{p,q} := \{p + t \cdot (q-p) ~ | ~ t \in {\mathbb R}, 0\le t\le 1\}$ be the {\em straight line} between $p$ and $q$. Embedding $\rho$ is called a {\em planar straight line embedding} for $G$, if for every pair of two distinct edges $\{u,v\},\{u',v'\} \in E$ the straight lines \[\overline{\rho(u),\rho(v)} \quad \text{and} \quad \overline{\rho(u'),\rho(v')}\] do not intersect except at a common endpoint.
\end{definition}

%\begin{definition}\label{GabrielGraph}
%An undirected graph $G=(V,E)$ is a {\em Gabriel graph}, if there is a {\em Gabriel embedding} $\rho: V \to \mathbb{R}^2$ such that all edges are Gabriel edges, where an edge $\{u,v\} \in E$ is a {\em Gabriel edge} if there is no vertex $w \in V$ incident with two edges $\{u,w\},\{w,v\} \in E$ such that
%\[\Vert \rho(u) - \rho(w) \Vert_2^2 + \Vert \rho(w) - \rho(v) \Vert_2^2 \quad \leq \quad \Vert \rho(u) - \rho(v) \Vert_2^2.\]

%An undirected graph $G=(V,E)$ is a {\em GUDG} if there is an embedding $\rho: V \to \mathbb{R}^2$ that is a Gabriel and a UDG embedding.
%\end{definition}

%Note that GUDGs are always planar and a GUDG embedding is a planar straight line embedding.

%%%%%%%%%%%%%%%%%%%%%%%%%%%%%%%%%%%%%%%%%%%%%%%%%%%%%%%%%%%%%%%%%%%%%%%%%%%%%%%%

\subsection{Metric Dimension}

The {\em distance} $d(u,v)$ between two nodes $u,v \in V$ in a graph $G=(V,E)$ is the smallest integer $k$ for which there is a path between $u$ and $v$ with $k$ edges. 

\begin{definition}
Let $G=(V,E)$ be an undirected graph. A set of vertices $S \subseteq V$ is a {\em resolving set} for $G$, see for example \cite{Bou09,CEJO00,HMPSCP05}, if for every pair $u,v \in V$ of two distinct vertices there is a vertex $s \in S$ such that $d(u,s) \not= d(v,s)$. The minimum size of a resolving set for an undirected graph $G$ is called the {\em metric dimension} of $G$.
\end{definition}

In this paper, we consider the complexity of the following decision problem.

\problemDecision{GUDG Metric Dimension}
{A GUDG $G=(V,E)$ and a positive integer $k$.}
{Is there a resolving set $S \subseteq V$ for $G$ of size at most $k$?}

\medskip
For a given undirected graph $G$ and a positive integer $k$, deciding whether there is a resolving set for $G$ of size at most $k$ is NP-complete for general graphs \cite{GJ79,KRR96} and even for planar graphs \cite{DPSL12}. We extend this and show that it remains NP-complete for GUDGs with vertex degree $\leq 6$.

%%%%%%%%%%%%%%%%%%%%%%%%%%%%%%%%%%%%%%%%%%%%%%%%%%%%%%%%%%%%%%%%%%%%%%%%%%%%%%%%

\subsection{Satisfiability}

Let $X=\{x_1,\ldots,x_n\}$ be a set of {\em boolean variables}. A {\em truth assignment} for $X$ is a function $t:X \to \{\text{true}, \text{false}\}$. If $t(x_i) = \text{true}$ then we say variable $x_i$ is {\em true} under $t$; if $t(x_i) = \text{false}$ then we say variable $x_i$ is {\em false} under $t$. If $x_i$ is a variable of $X$ then $x_i$ and $\overline{x_i}$ are {\em literals} over $X$; $x_i$ is called a {\em positive literal}, $\overline{x_i}$ is called a {\em negative literal}. The positive literal $x_i$ is true under $t$ if and only if variable $x_i$ is true under $t$; negative literal $\overline{x_i}$ is true under $t$ if and only if variable $x_i$ is false under $t$. A {\em clause} over $X$ is a set of literals over $X$. It represents the disjunction of literals which is {\em satisfied by a truth assignment} $t$ if and only if at least one of its literals is true under $t$. A collection ${\cal C}$ of clauses over $X$ is {\em satisfiable} if and only if there is a truth assignment $t$ that 
simultaneously satisfies all clauses of ${\cal C}$.

\begin{definition}
Let $X$ be a set of boolean variables and ${\cal C}$ be a collection of clauses over $X$. The directed {\em clause variable graph} $G_{\psi}$ of $\psi = (X,{\cal C})$ has a vertex $x$ for every variable $x\in X$ and a vertex $c$ for every clause $c\in {\cal C}$.
There is a directed edge $(x,c)$ from {\em variable vertex} $x$ to {\em clause vertex} $c$ if and only if $c$ contains literal $x$ or $\overline{x}$.
\end{definition}

%In \cite{DJPSY94}, the NP-completeness of the following problem is shown.

%\problemDecision{Restricted Planar 3-Sat}
%{A set $X$ of boolean variables and a collection ${\cal C}$ of clauses over $X$ such that
%\begin{compactitem}
%\item every variable $x\in X$ occurs in exactly three clauses, once as a negative literal and twice as a positive literal,
%\item the clause variable graph $G_{\psi}$ for $\psi = (X,{\cal C})$ is planar and
%\item every clause contains two or three literals.
%\end{compactitem}
%}
%{Is there a satisfying truth assignment for ${\cal C}$?}

%We need the following extended version of {\sc Restricted Planar 3-Sat} introduced in \cite{DPSL12}. The differences are emphasized by cursive script.

The following problem is NP-complete as shown in Corollary 3 of \cite{DPSL12}. 

\problemDecision{$1$-Negative Planar 3-Sat}
{A set $X$ of boolean variables and a collection ${\cal C}$ of clauses over $X$ such that
\begin{compactitem}
\item every variable $x \in X$ occurs in exactly two or three clauses, once as a negative literal, and once or twice as a positive literal,
\item the clause variable graph $G_{\psi}$ for $\psi = (X,{\cal C})$ is planar,
\item every clause contains two or three literals and
\item every clause with three literals contains at least one negative literal.
\end{compactitem}
}
{Is there a satisfying truth assignment for ${\cal C}$?}

\begin{definition}\label{PlanarOrthogonalGridDrawing}
Let $G=(V,E)$ be a directed planar graph and $\rho: V \to \mathbb{Z}^2$ be an embedding of the vertices of $G$. An {\em edge path} of {\em length $k-1$} for an edge $(u,v) \in E$ is a sequence of $k$ points
$\ep{u}{v} := (p_1, \dots , p_k), \quad p_i \in \mathbb{Z}^2, \quad 1\le i\le k,$
for which %1. $(i \not= j) \Rightarrow (p_i \not= p_j)$, 2. $p_1 = \rho(u)$, $p_k = \rho(v)$, and 3. for $1\le i < k$, if $p_i = (x_i,y_i)$ then $p_{i+1}\in\left\{(x_i+1,y_i),(x_i-1,y_i), (x_i,y_i+1), (x_i,y_i-1)\right\}$.
\begin{enumerate}
\item $(i \not= j) \Rightarrow (p_i \not= p_j)$,
\item $p_1 = \rho(u)$, $p_k = \rho(v)$, and
\item for $1\le i < k$, if $p_i = (x_i,y_i)$ then \\ $p_{i+1}\in\left\{(x_i+1,y_i),(x_i-1,y_i), (x_i,y_i+1), (x_i,y_i-1)\right\}$.
\end{enumerate}

%Every edge path $\ep{u}{v}$ represents a sequence of connected straight unit length lines $\overline{p_i,p_{i+1}}$ in the Euclidean plane $\mathbb{R}^2$. These lines are called {\em edge path segments}. Every edge path segment is a line of length $1$ parallel to either the $x$ or the $y$ axis of the coordinate system.

A {\em planar orthogonal grid drawing} of a directed graph $G=(V,E)$ is a pair $(\rho,{\cal E})$ where
\begin{itemize}
\item
$\rho: V \to \mathbb{Z}^2$ is an embedding of the vertices of $G$,
\item
${\cal E}$ is a collection of edge paths, one for every edge of $G$, and
%\item no two distinct edge paths $\ep{u}{v}$ and $\ep{u'}{v'}$ of ${\cal E}$ have a common point except $\rho(u)$ or $\rho(v)$ if the corresponding two edges $\{u,v\}$ and $\{u',v'\}$ have a common vertex.
\item no two distinct edge paths $\ep{u}{v}$ and $\ep{u'}{v'}$ of ${\cal E}$ have a common point $p$ unless the corresponding two edges $(u,v)$ and $(u',v')$ have a common vertex $w \in \{u,v\} \cap \{u',v'\}$ and $p=\rho(w)$.
\end{itemize}
\end{definition}

%%%%%%%%%%%%%%%%%%%%%%%%%%%%%%%%%%%%%%%%%%%%%%%%%%%%%%%%%%%%%%%%%%%%%%%%%%%%%%%%

\section{Main result}

%The main result of this paper is the following theorem.

\begin{theorem}\label{PlanarUDGMetricDimension}
{\sc GUDG Metric Dimension} is NP-complete.
\end{theorem}

The membership to NP is obvious, because it is easy to verify in polynomial time whether a given set of vertices $S$ is a resolving set. The NP-hardness is shown by a reduction from {\sc $1$-Negative Planar 3-Sat}.%, which is NP-complete as shown in Theorem \ref{negPlanar3sat}. The transformation constructs in polynomial time for an arbitrary instance $\psi = (X,{\cal C})$ for {\sc $1$-Negative Planar 3-Sat} a GUDG $H_{\psi}=(V_{H_{\psi}},E_{H_{\psi}})$ and an integer $k_{\psi}$ such that there is a truth assignment for $C$ if and only if there is a resolving set $W \subseteq V_{H_{\psi}}$ for $H_{\psi}$ of size at most $k_{\psi}$.
%Let $G_{\psi}$ be a planar clause variable graph for $\psi$. Every directed planar graph with vertex degree at most four has a planar orthogonal grid drawing. Such a planar orthogonal grid drawing can be computed in linear time, see \cite{TT89}. Since the vertices of $G_{\psi}$ have a degree of at most three, we can assume that we have a planar orthogonal grid drawing $(\rho,{\cal E})$ for $G_{\psi}$ whose size is polynomially bounded in the size of $\psi$, see \cite{Tam87}. The size of ${\cal E}$ is determined by the number of points $p_i \in {\mathbb R}^2$ of the edge paths $\overline{(u,v)}$ of ${\cal E}$.
Every directed planar graph with vertex degree at most four has a planar orthogonal grid drawing, whose size (determined by the number of grid points used for the edge paths) is polynomially bounded in the number of vertices. Such a drawing can be computed in polynomial time using the algorithms described in \cite{Tam87,TT89}. Let $G_{\psi}$ be a planar clause variable graph for $\psi$. Since the vertices of $G_{\psi}$ have a degree of at most three, we can assume that we have a planar orthogonal grid drawing $(\rho,{\cal E})$ for $G_{\psi}$.

%The following preprocessing phase modifies $\psi=(X,{\cal C})$ into an equivalent instance $\psi'=(X',{\cal C}')$ such that there is a truth assignment for $C$ if and only if there is a truth assignment for $C'$ and there is a planar orthogonal grid drawing $(\rho',{\cal E}')$ for the clause variable graph $G_{\psi'}$ whose edge paths have a length of at most two. This modification repeatedly splits long edge paths into shorter ones. Every splitting step inserts an additional variable similar to the proof of Theorem \ref{negPlanar3sat}.

The following preprocessing phase modifies $\psi=(X,{\cal C})$ into an equivalent instance $\psi'=(X',{\cal C}')$ for which there is a planar orthogonal grid drawing $(\rho',{\cal E}')$ for the clause variable graph $G_{\psi'}$ whose edge paths have a length of at most two. This upper bound on the length of the edge paths is necessary, because the gadget that is going to represent the edge paths cannot be embedded across arbitrarily large areas.

Let $(x,c)$ be an edge of the clause variable graph $G_{\psi}$ whose edge path in ${\cal E}$ is $\ep{x}{c}=(p_1, \dots , p_k)$ with length $l = k-1 \ge 3$. Assume clause $c$ contains a positive literal $x$ (the case for a negative literal $\overline{x}$ runs analogously). Then a new variable $h$ and a new clause $c'=\{x,\overline{h}\}$ are inserted and the positive literal $x$ in clause $c$ is replaced by the positive literal $h$. The new sets of variables and clauses again define an instance for {\sc $1$-Negative Planar 3-Sat} that has a satisfying truth assignment if and only if the original instance has a satisfying truth assignment. The original planar orthogonal grid drawing for $G_{\psi}$ is modified for the new clause variable graph as follows. Define $\rho(c') := p_2$, $\rho(h):=p_3$, $\ep{x}{c'} := (p_1,p_2)$, $\ep{h}{c'} := (p_3,p_2)$, and $\ep{h}{c} := (p_3,\ldots,p_k)$. The old edge path of length $l$ is now replaced by two edge paths of length $1$ and one edge path of length $l-2$. This 
splitting step is repeated until all edge paths have length at most two. See Figure \ref{edge-path-shortening} for an example.

\begin{figure}[hbt]
\medskip
\center
\includegraphics[width=340pt]{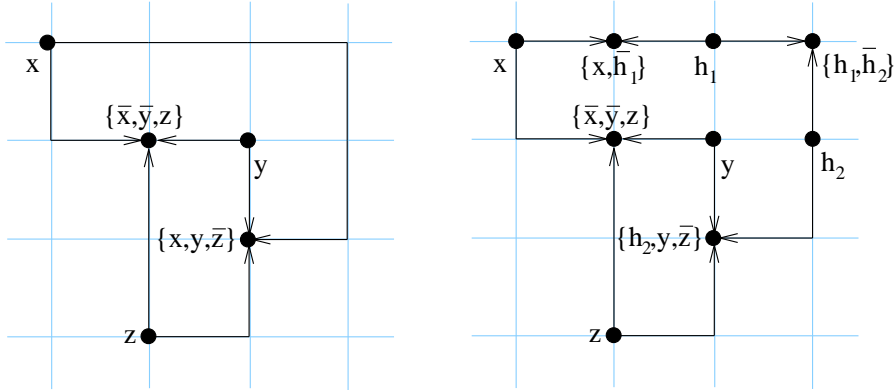}
\caption{Left: A planar orthogonal grid drawing of the clause variable graph $G_{\psi}$ for $\psi = (X,{\cal C})$ with $X=\{x,y,z\}$ and ${\cal C} = \{ \{x,y,\overline{z}\}, \{\overline{x}, \overline{y}, z\} \}$; Right: A planar orthogonal grid drawing of $G_{\psi'}$\label{edge-path-shortening}}
\end{figure}

%%%%%%%%%%%%%%%%%%%%%%%%%%%%%%%%%%%%%%%%%%%%%%%%%%%%%%%%%%%%%%%%%%%%%%%%%%%%%%%%

\subsection{The construction of $H_{\psi}$}

The GUDG $H_{\psi}$ is assembled using copies of the following 6 graphs called {\em gadgets}. Some of these gadgets are similar to the gadgets used in \cite{DPSL12} for the proof that {\sc Planar Metric Dimension} is NP-complete. The modifications are done to obtain GUDGs. Every gadget has two or three special path pairs called {\em $(t,f)$-path pairs}. The $(t,f)$-path pairs are used to represent the edges of the clause variable graph $G_{\psi'}$. The end-vertices of degree one of a $(t,f)$-path pair are called a {\em $(t,f)$-vertex pair}.

\begin{enumerate}
\item
There are three {\em variable gadgets} $G^3_{v}, G^2_{v,a}$, and $G^2_{v,b}$ for the variables of $X'$, see Figure \ref{variableGadgets}. The first variable gadget $G^3_{v}$ is used for the case that the corresponding variable is contained in three clauses, once as a negative literal and twice as a positive literal. The other two %variable gadgets $G^2_{v,a}$ and $G^2_{v,b}$ are induced subgraphs of $G^3_{v}$. They 
are used for the case that the corresponding variable is contained in two clauses, once as a negative literal and once as a positive literal.

%\medskip
These three gadgets have two different types of $(t,f)$-path pairs marked by $\ominus$ for the negative literal and $\oplus$ for the positive literals. The $\ominus$-$(t,f)$-vertex pair is $(t_{1,14},f_{1,14})$, the $\oplus$-$(t,f)$-vertex pairs are $(t_{2,14},f_{2,14})$ and $(t_{3,14},f_{3,14})$.

\item
There are two {\em clause gadgets} $G^3_{c}, G^2_{c}$ for the clauses of ${\cal C}'$, see Figure \ref{clauseGadgets} to the left and in the middle. The first clause gadget $G^3_{c}$ is used for clauses with three literals. The second clause gadget $G^2_{c}$ is an induced subgraph of $G^3_{c}$. It is used for clauses with two literals.

%\medskip
The $(t,f)$-vertex pairs are $(t_{1,15},f_{1,15})$, $(t_{2,15},f_{2,15})$ and $(t_{3,15},f_{3,15})$.

\item
There is an {\em edge gadget} $G_e$ for the edges of $G_{\psi'}$, see Figure \ref{clauseGadgets} to the right. It is used to connect variable gadgets with clause gadgets and consists of two disjoint paths with 37 vertices each.

%\medskip
These two paths represent one $(t,f)$-path pair with two $(t,f)$-vertex pairs $(t_{1},f_{1})$ and $(t_{37},f_{37})$.

%\begin{figure}[hbt]
%\center
%\includegraphics[width=160pt]{images/edge-gadget.eps}
%\caption{The edge gadget $G_{e}$. Every dashed line represents a path with 33 further vertices.}
%\label{edgeGadget}
%\end{figure}

\end{enumerate}

\begin{figure}[hbt]
\null
\hfill
\includegraphics[height=170pt]{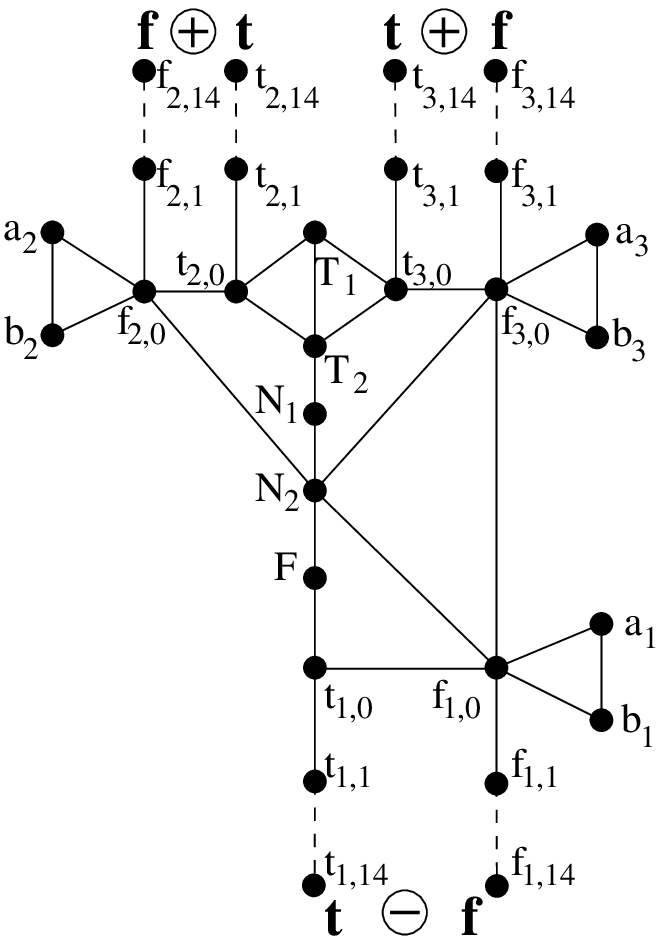}
\hfill
\includegraphics[height=170pt]{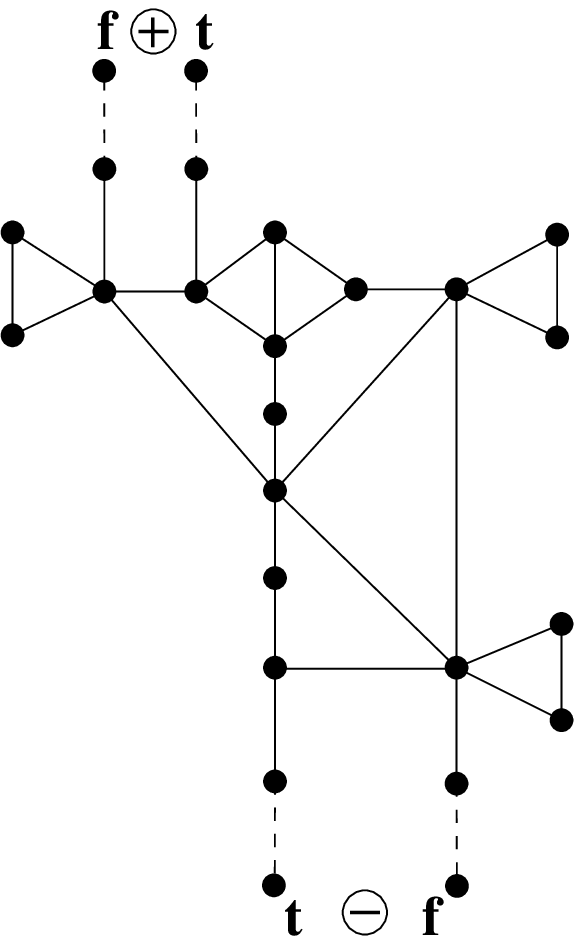}
\hfill
\includegraphics[height=170pt]{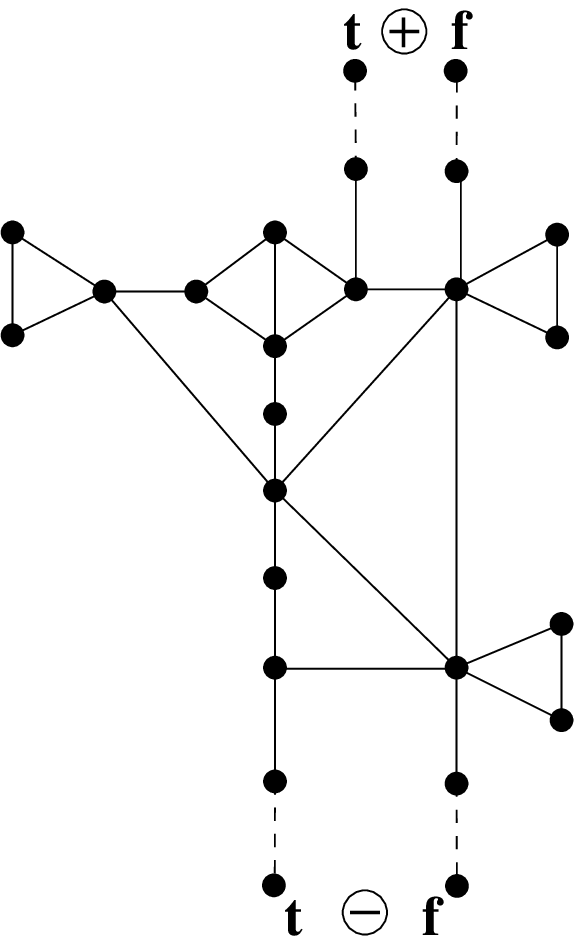}
\hfill
\null
\caption{The variable gadgets $G^3_{v}$, $G^2_{v,a}$, and $G^2_{v,b}$ f.l.t.r. The naming of the vertices in $G^2_{v,a}$ and $G^2_{v,b}$ is the same as in $G^3_{v}$, except that some vertices are missing. Every dashed line represents a path with 12 further vertices.}
\label{variableGadgets}
\end{figure}

\begin{figure}[hbt]
\null\hfill
\includegraphics[height=150pt]{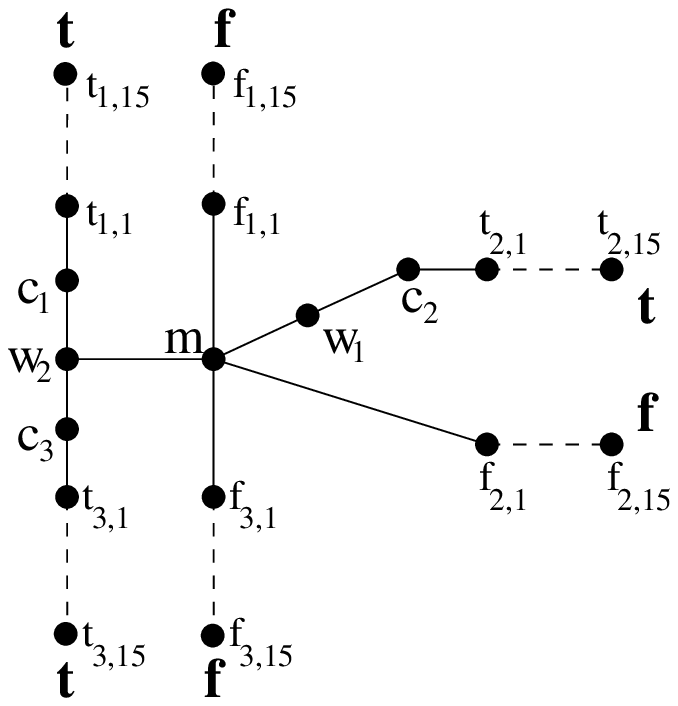} \hfill
\includegraphics[height=150pt]{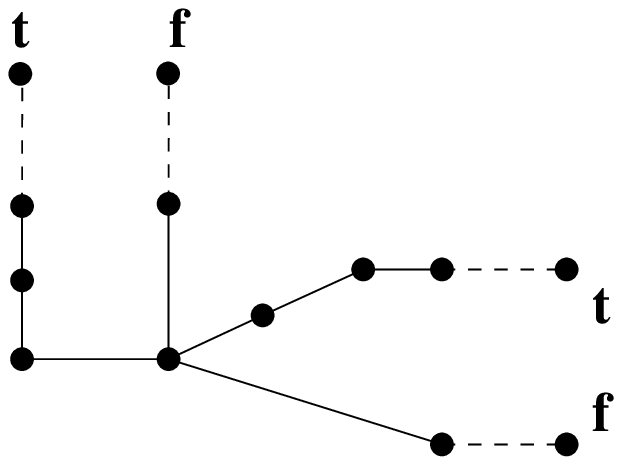} \hfill\null

\hfill\includegraphics[width=150pt]{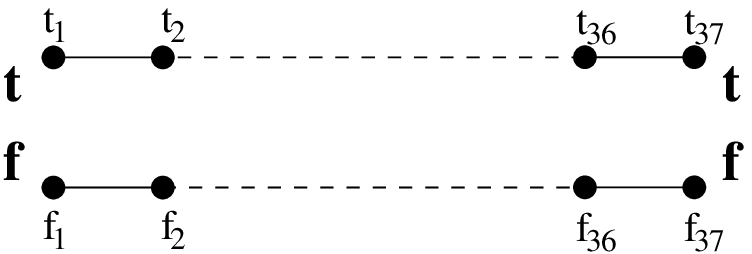} \hfill
\null
\caption{The clause gadgets $G^3_{c}$ and $G^2_{c}$ and the edge gadget $G_{e}$. The naming of the vertices in $G^2_{c}$ is the same as in $G^3_{c}$, except that some vertices are missing. Every dashed line represents a path with 13 or 33 further vertices, respectively.}
\label{clauseGadgets}
\medskip
\end{figure}

Graph $H_{\psi}$ contains for every variable $x \in X'$ with two positive literals and one negative literal in the clauses of ${\cal C}'$ a copy of $G^3_{v}$. It contains for every variable $x \in X'$ with one positive and one negative literal in the clauses of ${\cal C}'$ a copy of $G^2_{v,a}$ or $G^2_{v,b}$. Then $H_{\psi}$ contains for every clause of ${\cal C}'$ with three literals one copy of $G^3_{c}$ and for every clause of ${\cal C}'$ with two literals one copy of $G^2_{c}$. Finally it contains for every edge of $G_{\psi'}$ one copy of $G_e$.

The gadgets are connected to each other by identifying $(t,f)$-vertex pairs as follows. Let $G_x$ be a variable gadget for variable $x$, $G_c$ a clause gadget for clause $c$, and $G_e$ the edge gadget for the edge $(x,c) \in E_{G_{\psi'}}$. One $(t,f)$-vertex pair of $G_e$ is identified with one $(t,f)$-vertex pair of $G_c$ and the other $(t,f)$-vertex pair of $G_e$ is identified with one $(t,f)$-vertex pair of $G_x$. If clause $c$ contains the positive literal $x$, then a $\oplus$-$(t,f)$-vertex pair is used, otherwise, a $\ominus$-$(t,f)$-vertex pair is used.
%the $\oplus$-$(t,f)$-vertex pair $(t_{2,14}, f_{2,14})$ or $(t_{3,14}, f_{3,14})$ of $G_x$ is used for the identification, if clause $c$ contains the negative literal $\overline{x}$, then the $\ominus$-$(t,f)$-vertex pair $(t_{1,14}, f_{1,14})$ of $G_x$ is used for the identification.

%For the case that clause $c$ contains the positive literal $x$, it does not matter whether we use $(t,f)$-vertex pair $(t_{2,14}, f_{2,14})$ or $(t_{3,14}, f_{3,14})$ of $G_c$. However, if $(t,f)$-vertex pair $(t_{2,14}, f_{2,14})$ is used for the positive literal $x$ then the other $(t,f)$-vertex pair $(t_{3,14}, f_{3,14})$ is used for the second positive literal of $c$, and vice versa.

Identifying two $(t,f)$-vertex pairs $(t_1,f_1)$ and $(t_2,f_2)$ means that all vertices adjacent to $t_2$ are connected to $t_1$, all vertices adjacent to $f_2$ are connected to $f_1$, and then the two vertices $t_2,f_2$ and their incident edges are removed.

%In the following sections we prove that $H_{\psi}$ can be assembled such that
%\begin{enumerate}
%\item
%it is a GUDG and
%\item
%there is a truth assignment for $C'$ if and only there is a resolving set $W \subseteq V_{H_{\psi}}$ for $H_{\psi}$ of size at most $k_{\psi} := 4 \cdot |X'|$.
%\end{enumerate}

%%%%%%%%%%%%%%%%%%%%%%%%%%%%%%%%%%%%%%%%%%%%%%%%%%%%%%%%%%%%%%%%%%%%%%%%%%%%%%%%

\subsection{A GUDG embedding for $H_{\psi}$}

In this section, it is shown that $H_{\psi}$ can be assembled such that there is a GUDG embedding $\rho_{H_{\psi}}$ for $H_{\psi}=(V_{\psi},E_{\psi})$. For each gadget $G$ we present several different GUDG embeddings. In every of these embeddings the vertices of the gadget are mapped to positions inside a polygon. The $(t,f)$-vertices are placed at the border of the polygon. The non-$(t,f)$-vertices, i.e. all vertices not belonging to a $(t,f)$-vertex pair, are placed at positions with a distance $> 1$ to all positions outside the polygon.

\begin{definition}\label{DefinitionTile}
A {\em tile} is a pair $(G,\rho_G)$, where $G=(V,E)$ is one of the gadgets defined above
%$G^3_{v}$, $G^2_{v,a}$, $G^2_{v,b}$, $G^3_{c}$, $G^2_{c}$, $G_{e}$ defined above
and $\rho_G:V \to \mathbb{R}^2$ is a GUDG embedding for $G$, called {\em tile embedding}, such that for every $(t,f)$-vertex pair $(t,f)$,
\[\{\rho_G(t),\rho_G(f)\} \in \left\{\begin{array}{l}\{(-6,-1),(-6,1)\}, ~ \{(6,-1),(6,1)\},\\\{(-1,-6),(1,-6)\}, ~ \{(-1,6),(1,6)\}\end{array}\right\}.\]

%For $a,b \in \mathbb{R}$, $a\le b$, let $[a,b]:=\{x\in\mathbb{R}~|~a\le x\le b\} \subset \mathbb{R}$ be the {\em closed interval} between $a$ and $b$. Then 
The $(t,f)$-vertices are placed at the border of the square $[-6,6] \times [-6,6]$ with side length $12$ and center $(0,0)$. %They are placed symmetrically around the middle of the border with a distance of $2$ from each other.

All other vertices of the gadget are mapped to positions of the square $[-6,6] \times [-6,6]$ that have a distance $>1$ to all positions of $\mathbb{R}^2 \setminus [-6,6] \times [-6,6]$ outside the square $[-6,6] \times [-6,6]$.
\end{definition}

We consider several different tile embeddings for gadget $G^3_{v}$. These are the three tile embeddings shown in Figure \ref{variableTiles3} and all tile embeddings obtained by $90$, $180$, or $270$ degree rotations and/or a horizontal or vertical mirroring. The rotations by $90$, $180$, and $270$ degree and the horizontal and vertical mirroring can be performed by the mappings
\[ 
\begin{array}{lllll}
f_{90^{\circ}} & : (x,y) \mapsto (-y,x), & $\qquad$ & f_{\text{horizontal}} & : (x,y) \mapsto (-x,y), \\ 
f_{180^{\circ}} & : (x,y) \mapsto (-x,-y), & $\qquad$ & f_{\text{vertical}} & : (x,y) \mapsto (x,-y), \\
f_{270^{\circ}} & : (x,y), \mapsto (y,-x). & & & \\
%f_{90^{\circ}} & : (x,y) \mapsto (-y,x), & $\quad$ & f_{180^{\circ}} & : (x,y) \mapsto (-x,-y), & $\quad$ & f_{270^{\circ}} & : (x,y), \mapsto (y,-x) \\
%f_{\text{horizontal}} & : (x,y) \mapsto (-x,y), & & f_{\text{vertical}} & : (x,y) \mapsto (x,-y). \\
\end{array}
\]

\begin{figure}[hbt]
\null\hfill
\includegraphics[width=100pt]{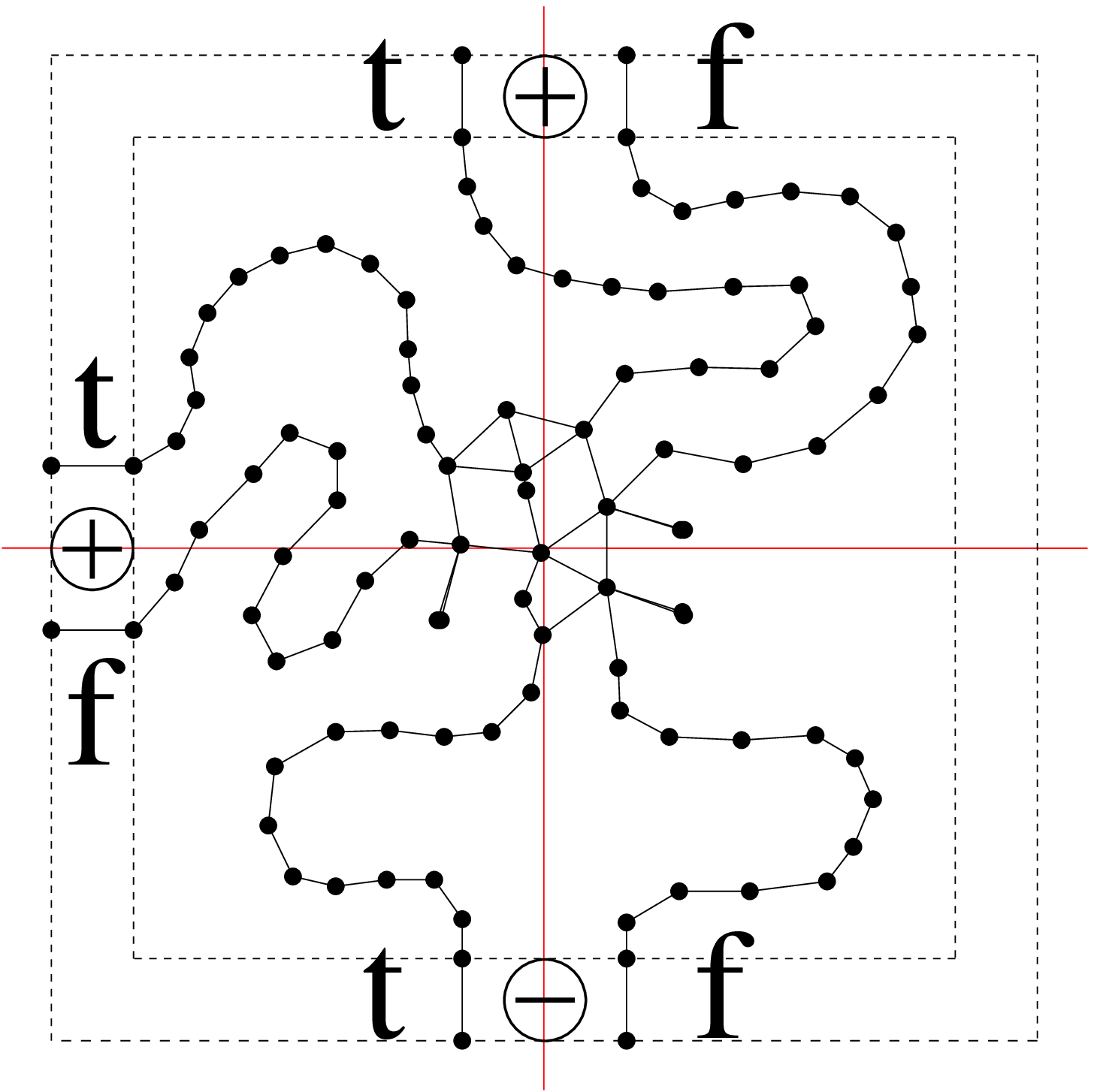} \hfill
\includegraphics[width=100pt]{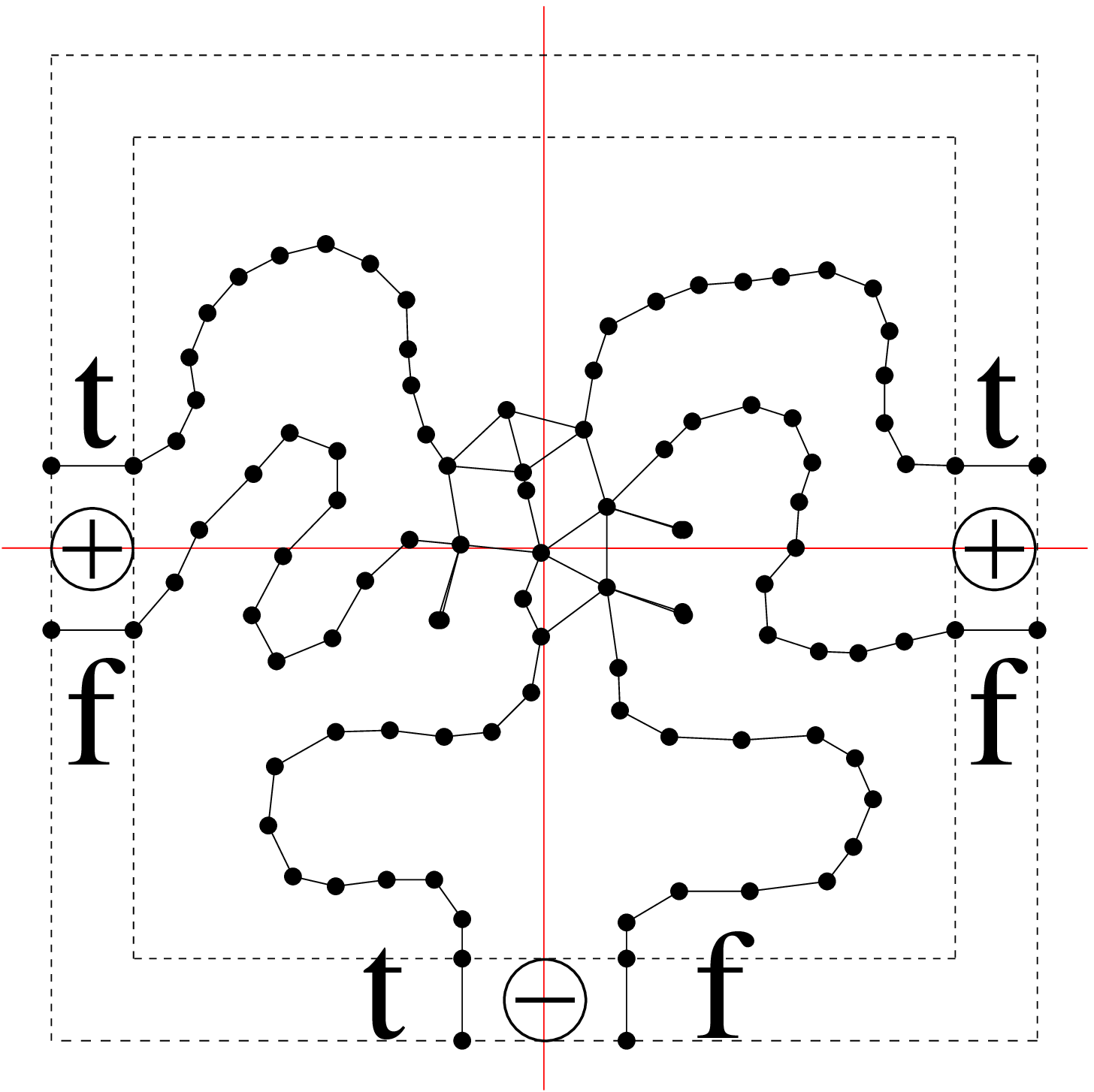} \hfill
\includegraphics[width=100pt]{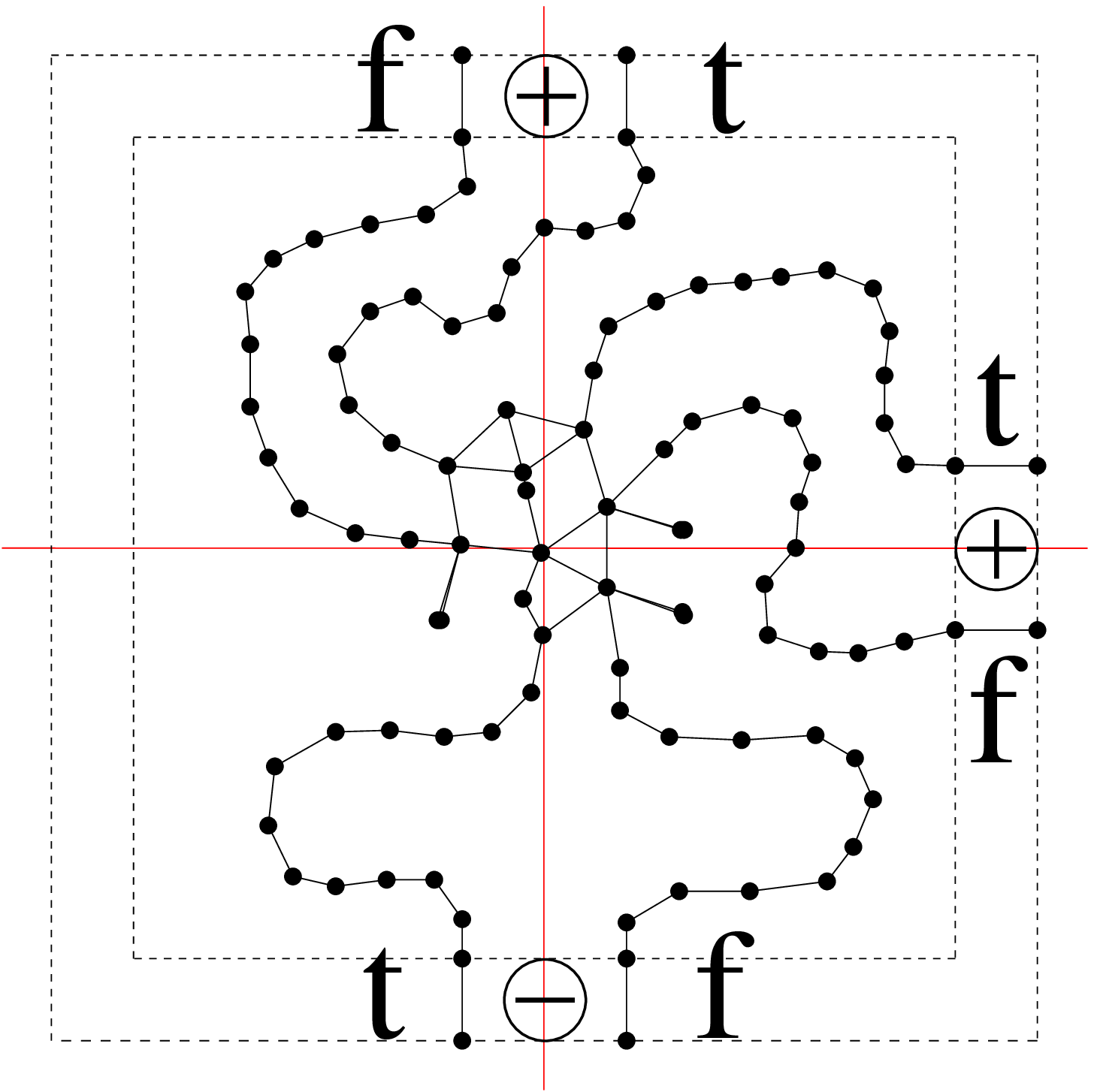} \hfill
\null
\caption{Three tile embeddings for variable gadget $G^3_{v}$. Rotating and mirroring yields further tile embeddings. The red $+$ marks the center $(0,0)$ of the coordinate system.}
\label{variableTiles3}
\medskip
\end{figure}

We also consider several different tile embeddings for every gadget $G^2_{v,a}$, $G^2_{v,b}$, $G^3_{c}$, $G^2_{c}$, and $G_{e}$, see Figures \ref{variableTiles2ab}, \ref{clauseTiles3var}, and \ref{clauseTiles2var}.

\begin{figure}[hbt]
\medskip
\null\hfill
\includegraphics[width=100pt]{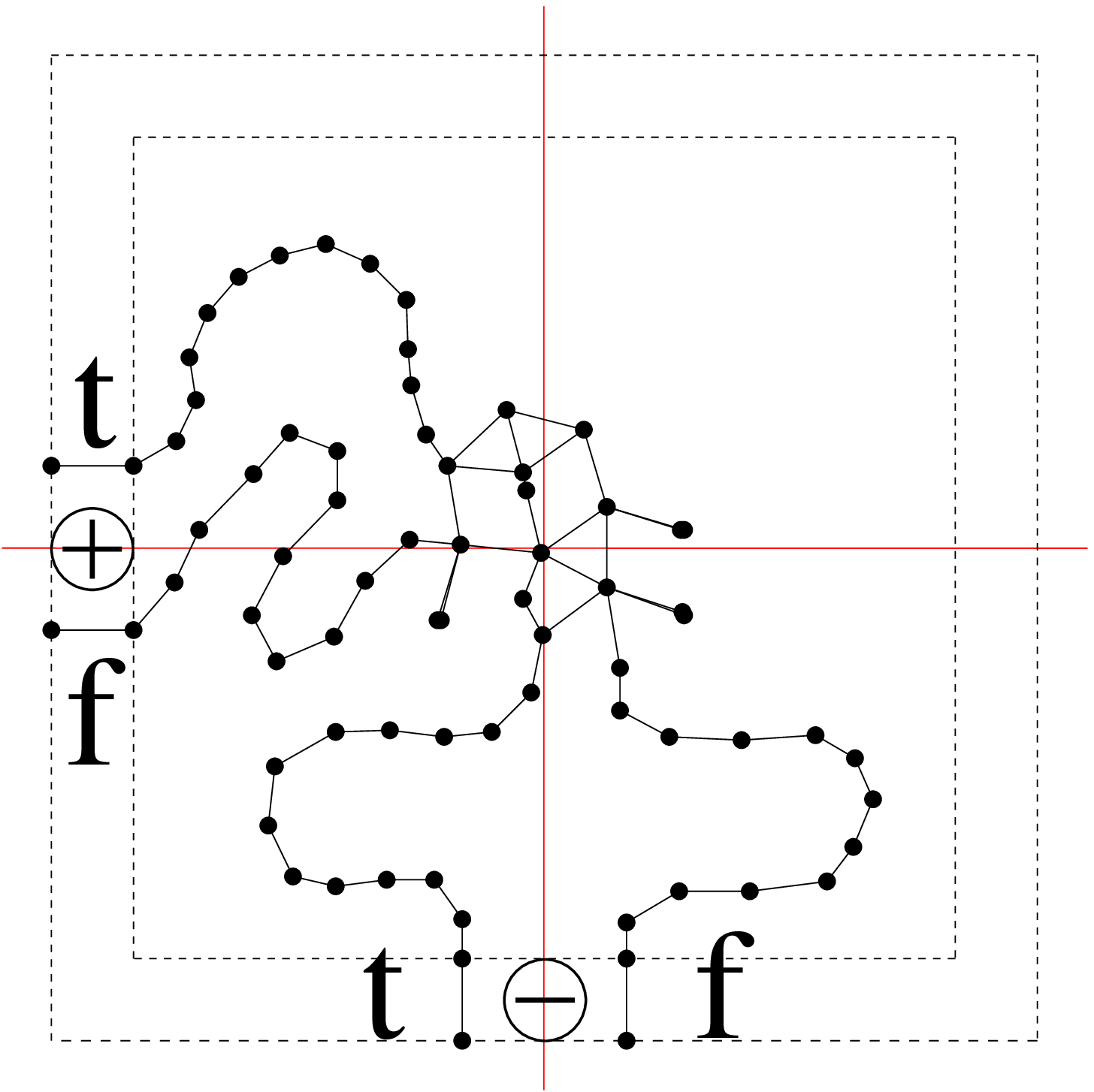} \hfill
\includegraphics[width=100pt]{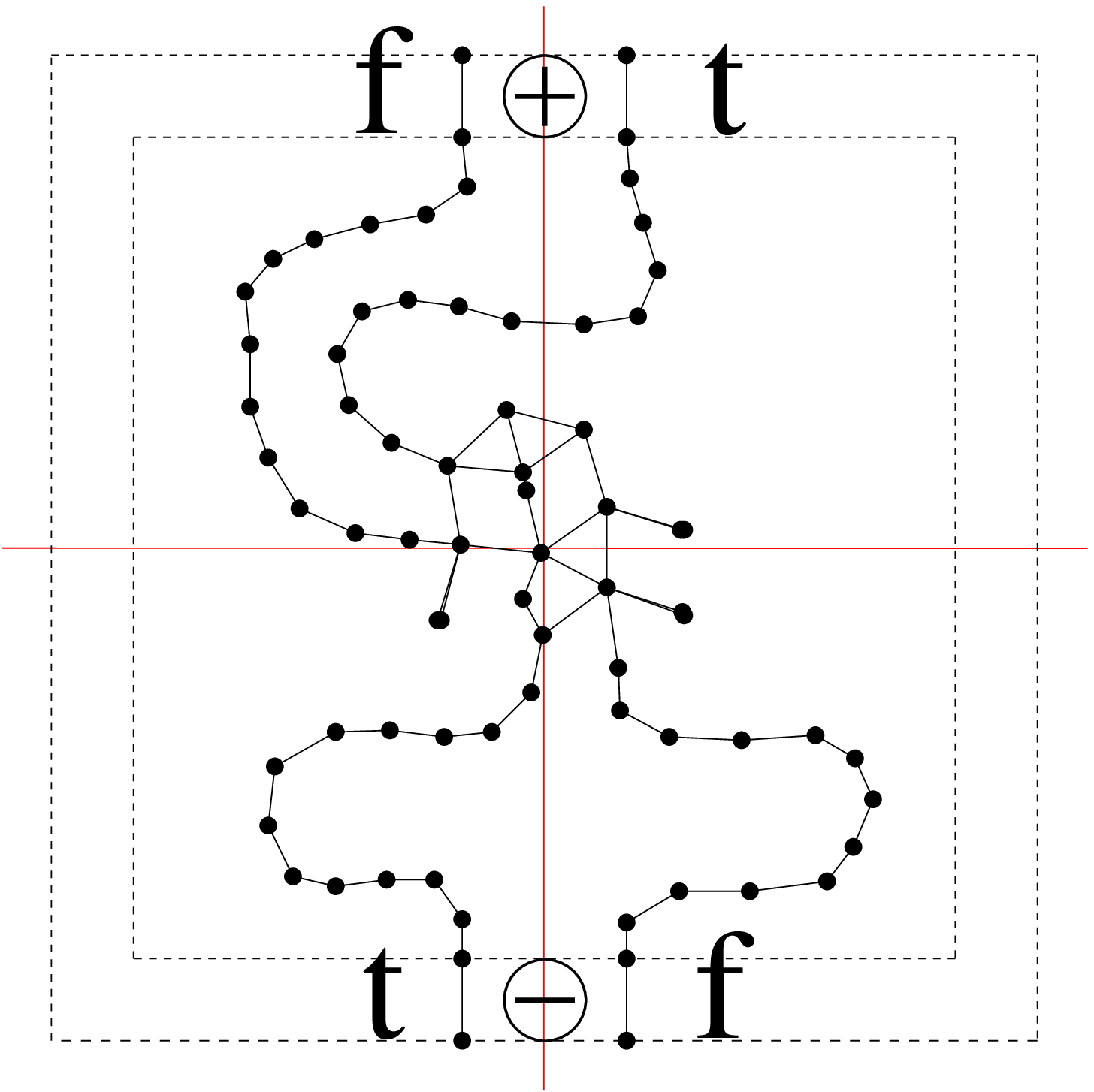} \hfill
\includegraphics[width=100pt]{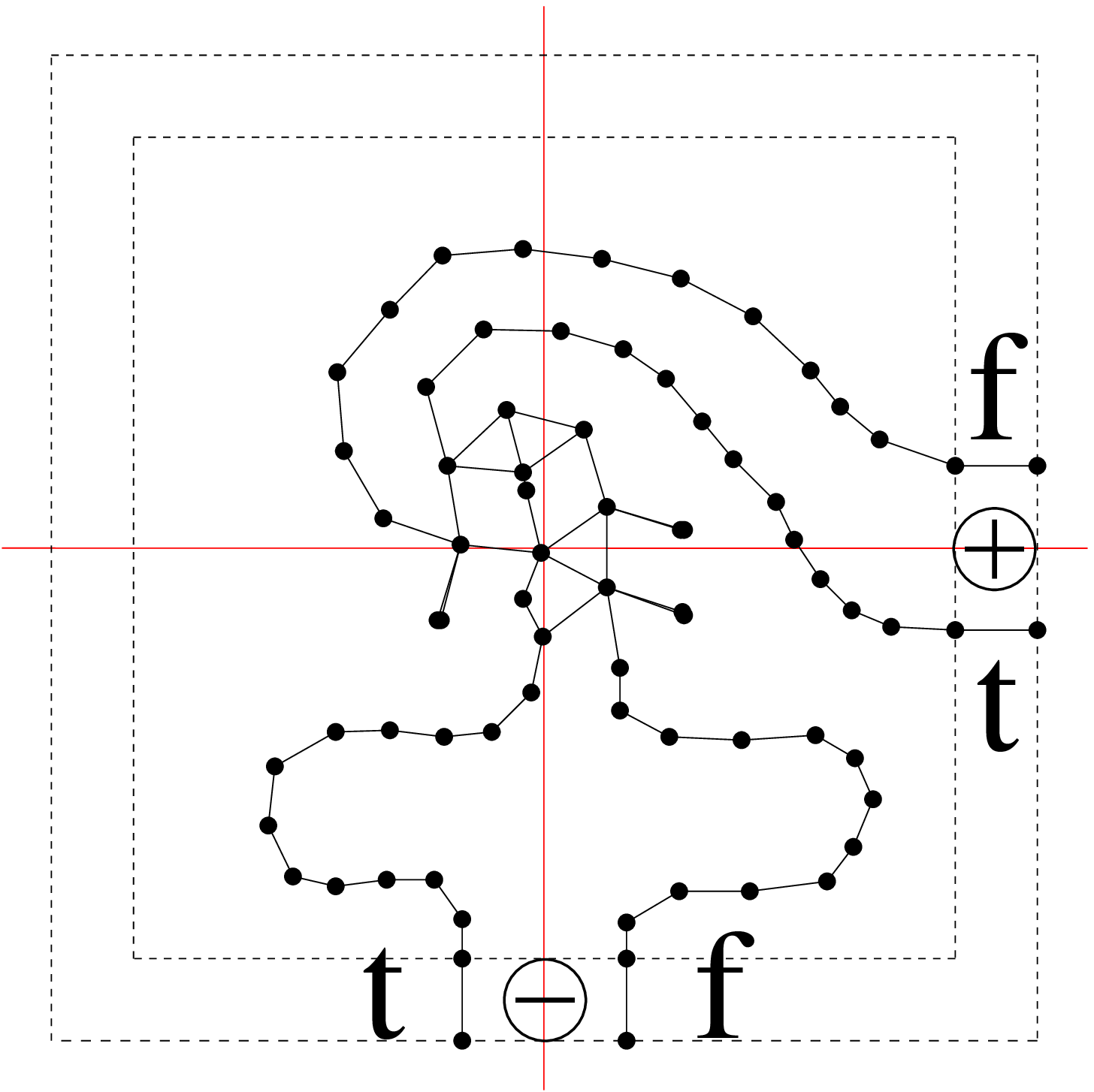} \hfill \null

\null\hfill
\includegraphics[width=100pt]{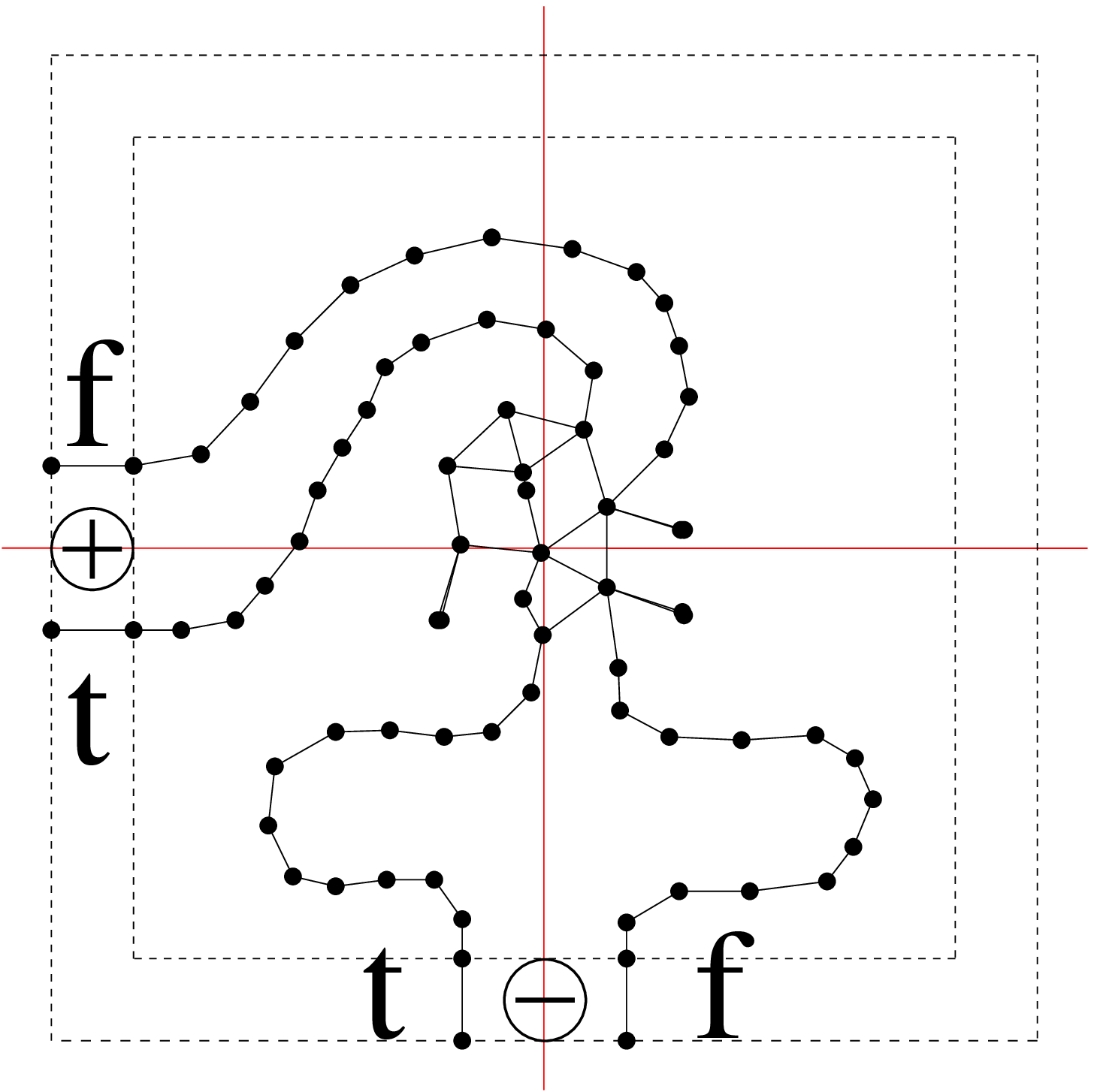} \hfill
\includegraphics[width=100pt]{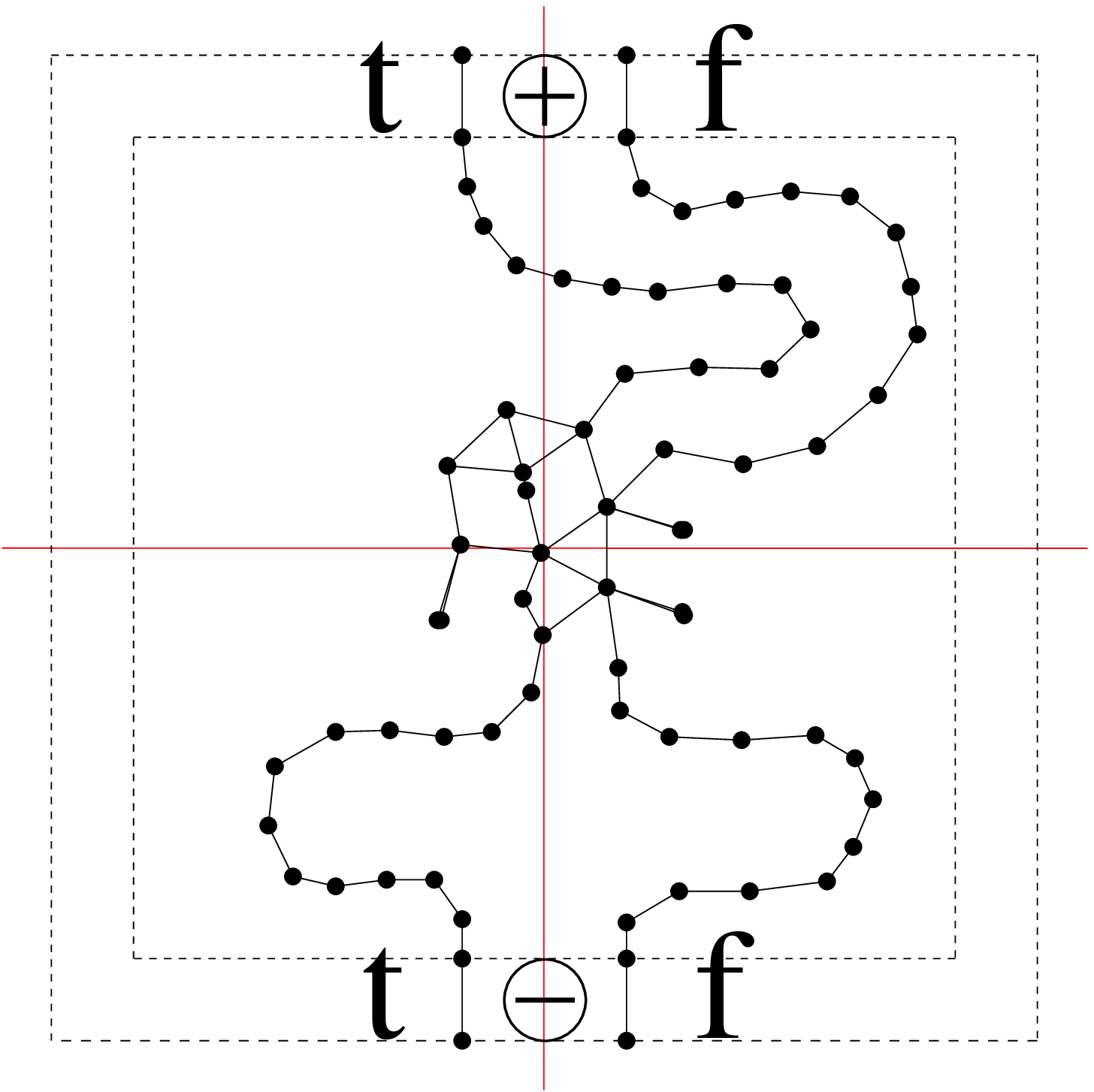} \hfill
\includegraphics[width=100pt]{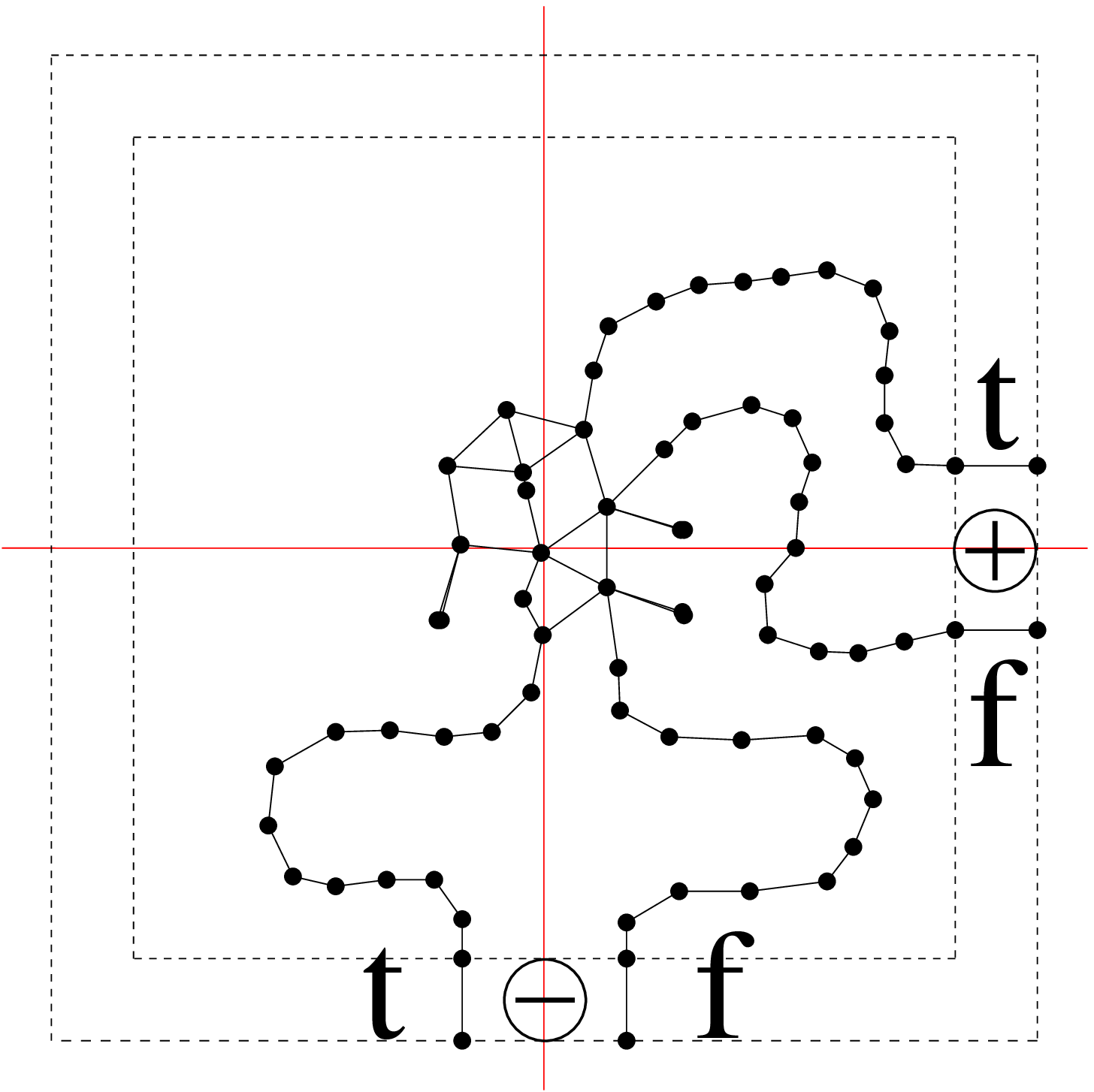} \hfill
\null
\caption{Three tile embeddings for variable gadget $G^2_{v,a}$ (top) and three tile embeddings for variable gadget $G^2_{v,b}$ (bottom). Rotating and mirroring yields further tile embeddings for each gadget.}
\label{variableTiles2ab}
\end{figure}

The tile embeddings resulting from the last embedding of Figure \ref{clauseTiles2var} for edge gadget $G_e$ are only used for edges whose edge paths have length 1. For edge paths of length 2, the following definition is used.

Let $S_{x,y} := \{(x'+x,y'+y) ~|~(x',y') \in [-6,6] \times [-6,6]\}$ be the square of side length $12$ and center $(0,0)$ translated by $(x,y) \in \mathbb{Z}^2$. The vertices of $G_e$ are mapped to positions of one of the following polygons $P_i$, $1 \leq i \leq 6$, formed by the union of $3$ squares.
\begin{center}
\begin{tabular}{llllllllllll}
$P_1 := $ & $S_{-12,0}$ & $\cup$ & $S_{0,0}$ & $\cup$ & $S_{12,0}$ & $\quad P_2 := $ & $S_{0,12}$  & $\cup$ & $S_{0,0}$ & $\cup$ & $S_{0,-12}$ \\
$P_3 := $ & $S_{-12,0}$ & $\cup$ & $S_{0,0}$ & $\cup$ & $S_{0,-12}$ & $\quad P_4 := $ & $S_{-12,0}$ & $\cup$ & $S_{0,0}$ & $\cup$ & $S_{0,12}$\\
$P_5 := $ & $S_{0,-12}$ & $\cup$ & $S_{0,0}$ & $\cup$ & $S_{12,0}$ & $\quad P_6 := $ & $S_{0,12}$  & $\cup$ & $S_{0,0}$ & $\cup$ & $S_{12,0}$ \\
\end{tabular}
\end{center}

\begin{figure}[hbt]
\null\hfill
\includegraphics[width=100pt]{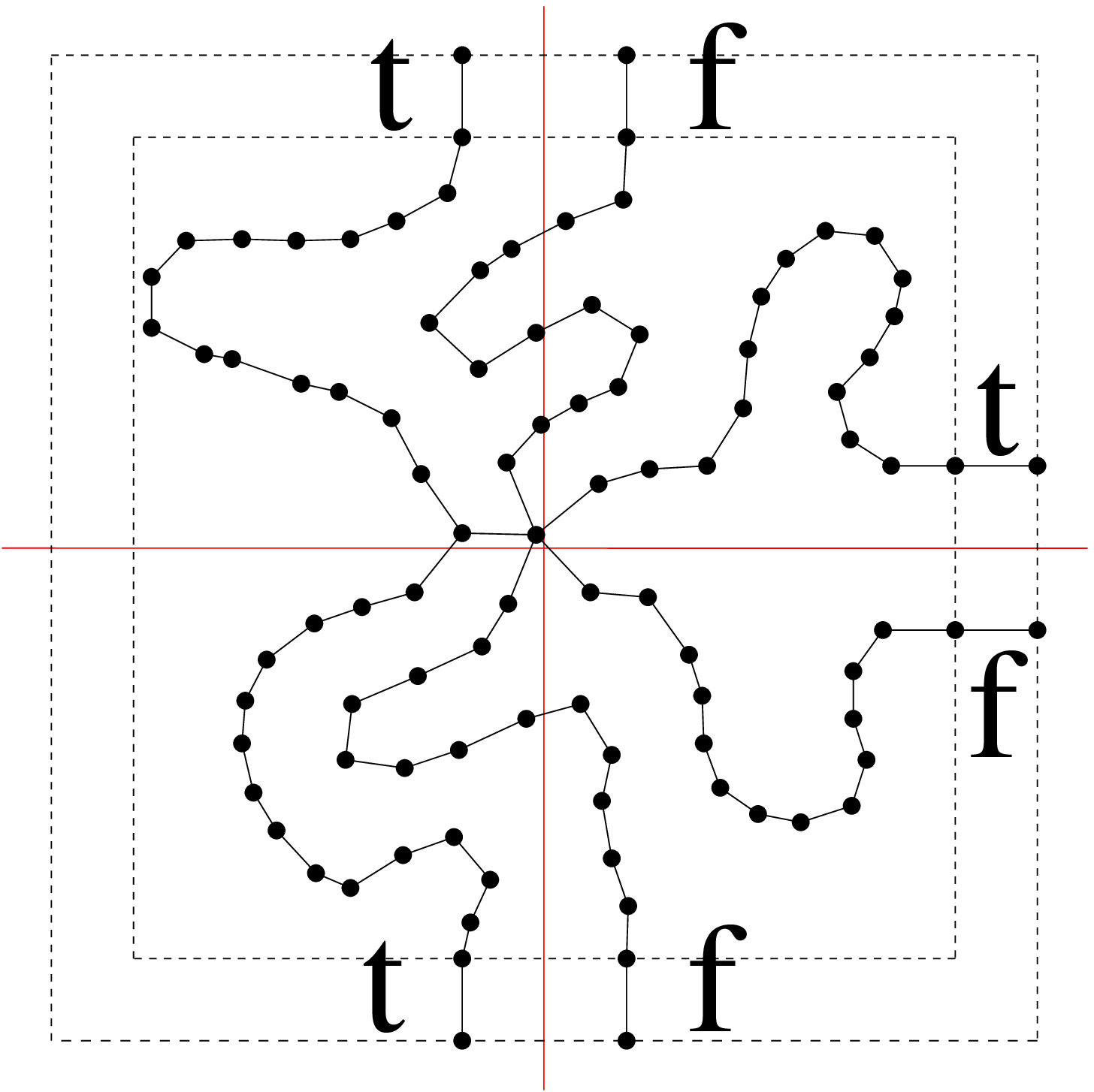} \hfill
\includegraphics[width=100pt]{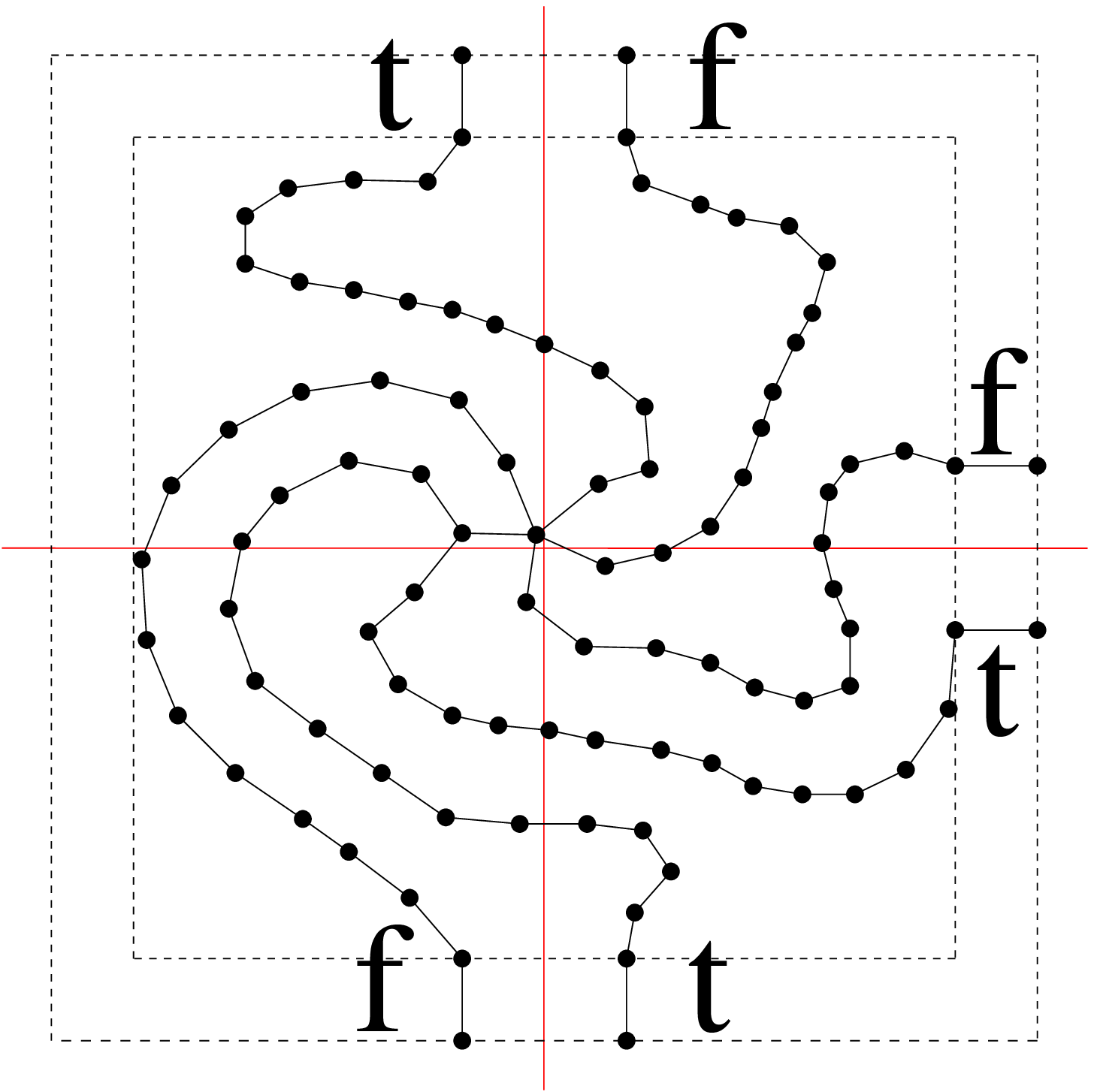} \hfill
\includegraphics[width=100pt]{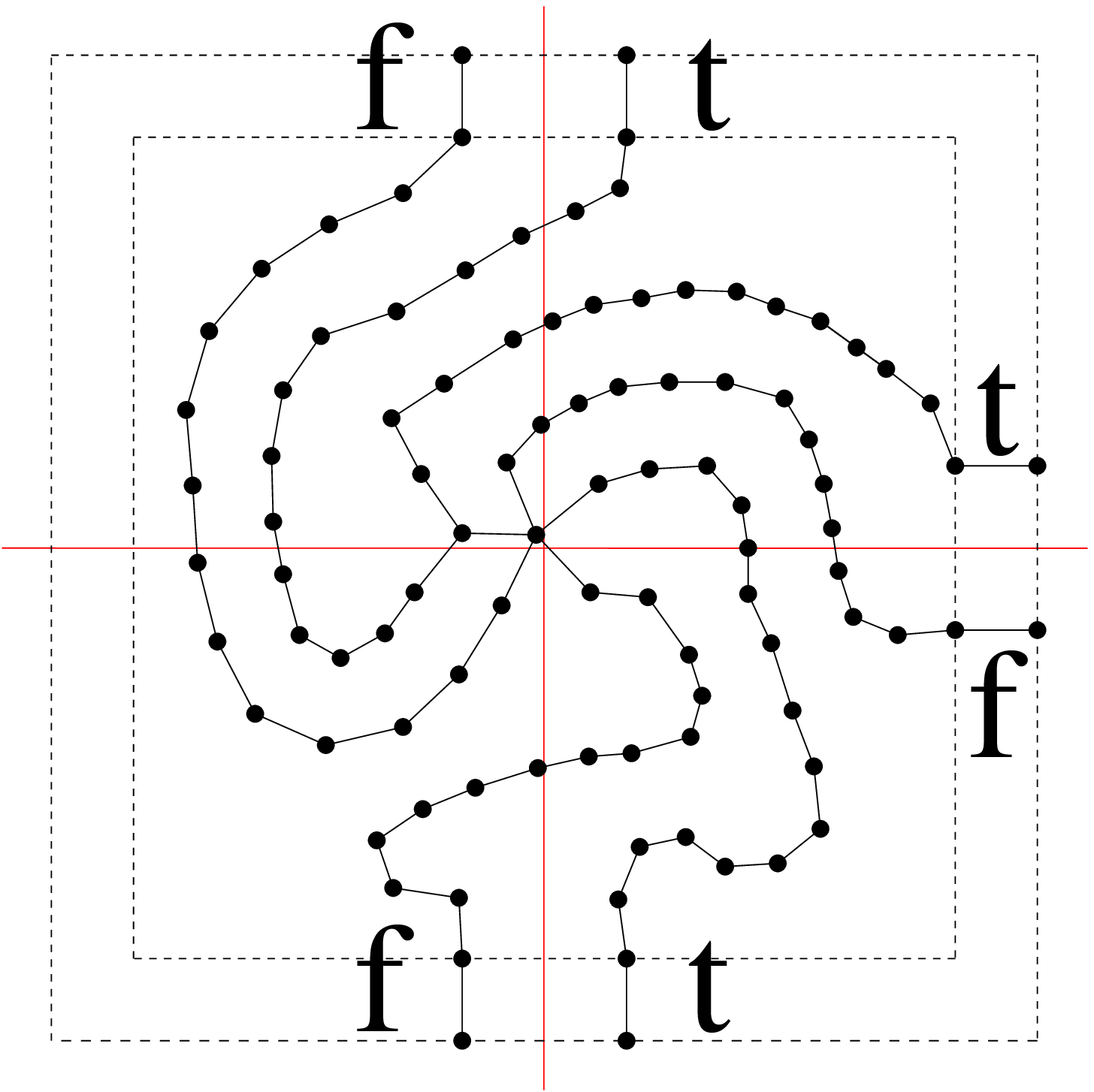} \hfill
\null
\caption{Three tile embeddings for clause gadget $G^3_{c}$. Rotating and mirroring yields further tile embeddings.}
\label{clauseTiles3var}
\medskip
\end{figure}

\begin{figure}[hbt]
\null\hfill
\includegraphics[width=100pt]{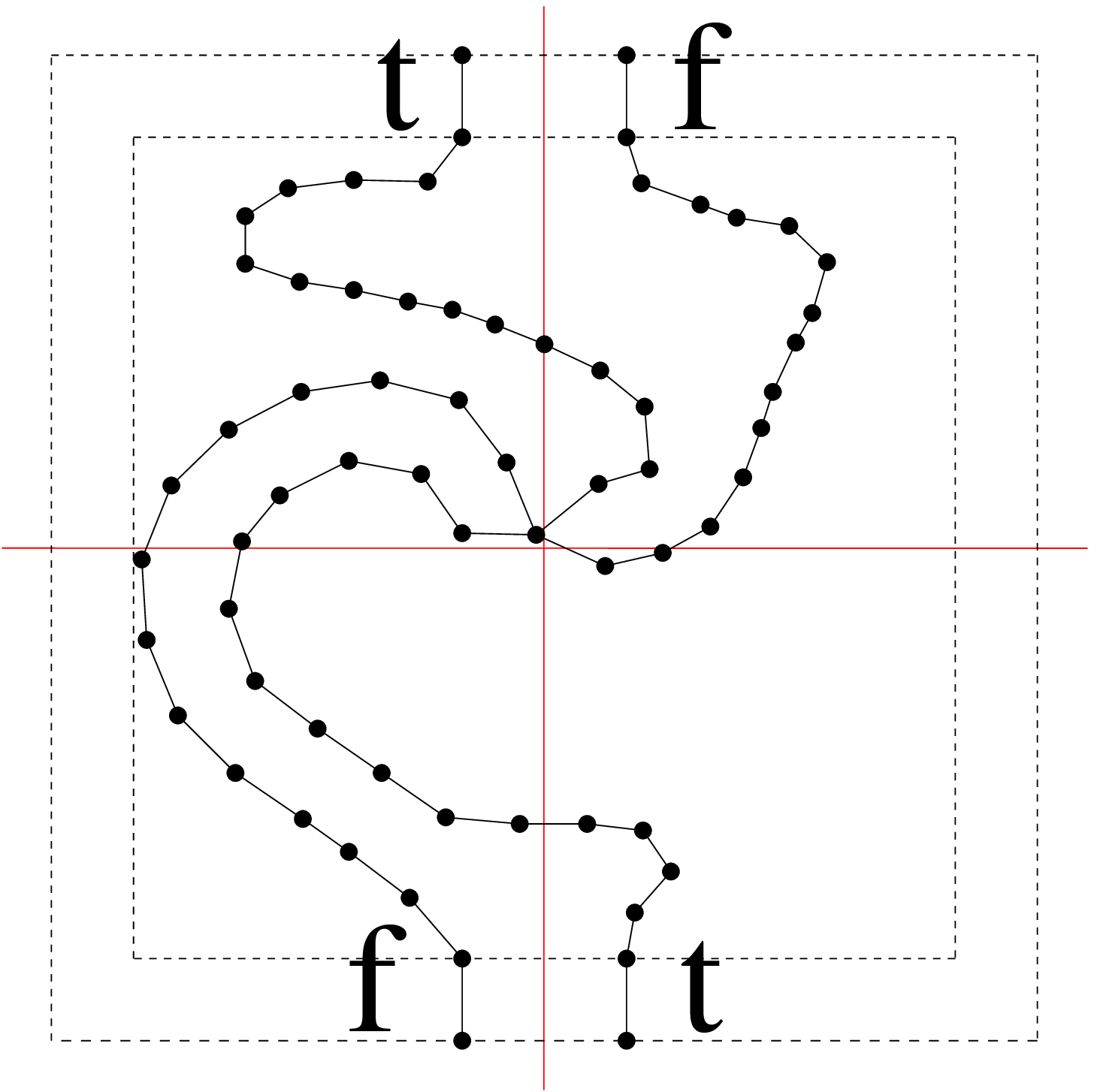} \hfill
\includegraphics[width=100pt]{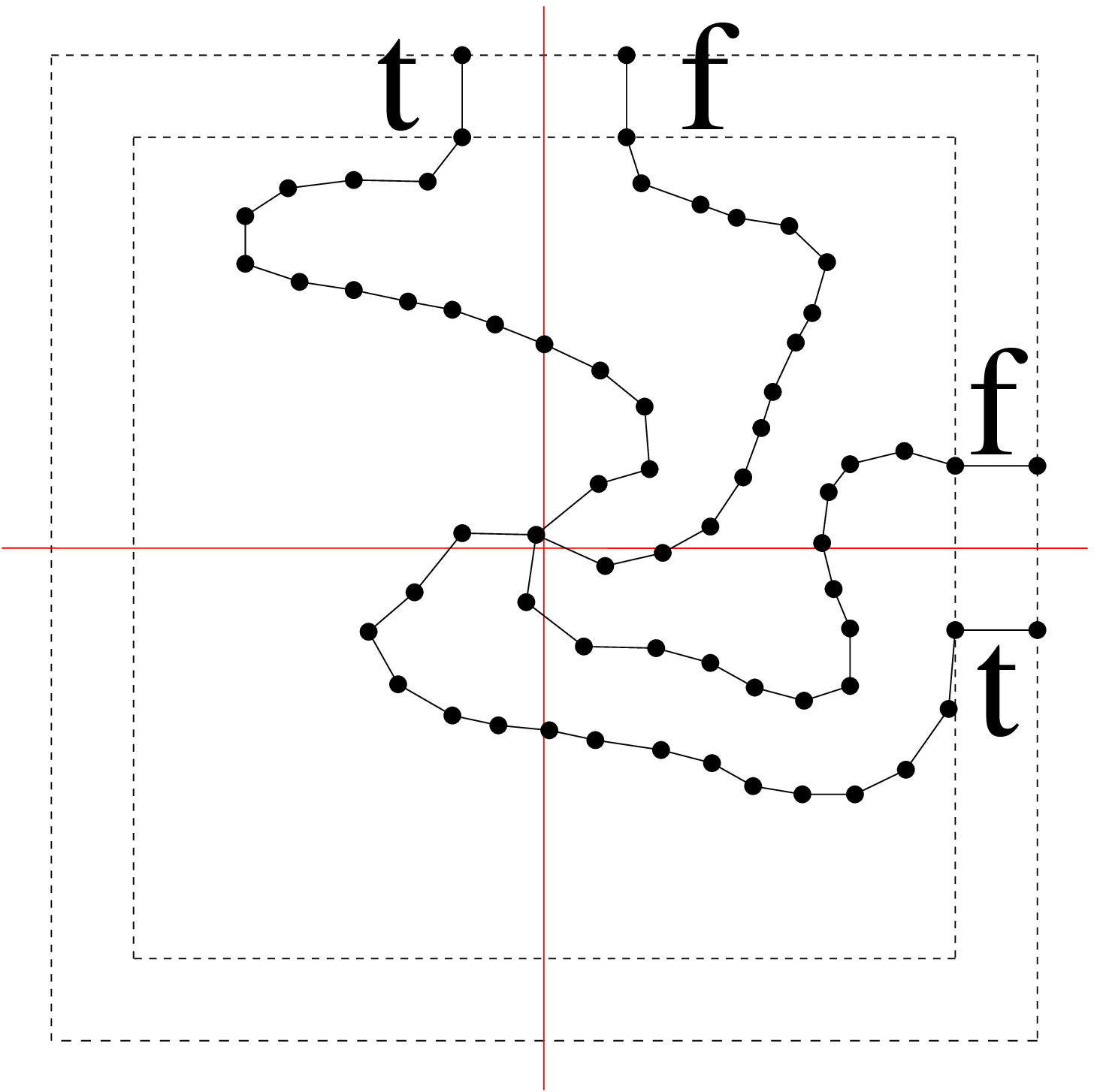} \hfill
\includegraphics[width=100pt]{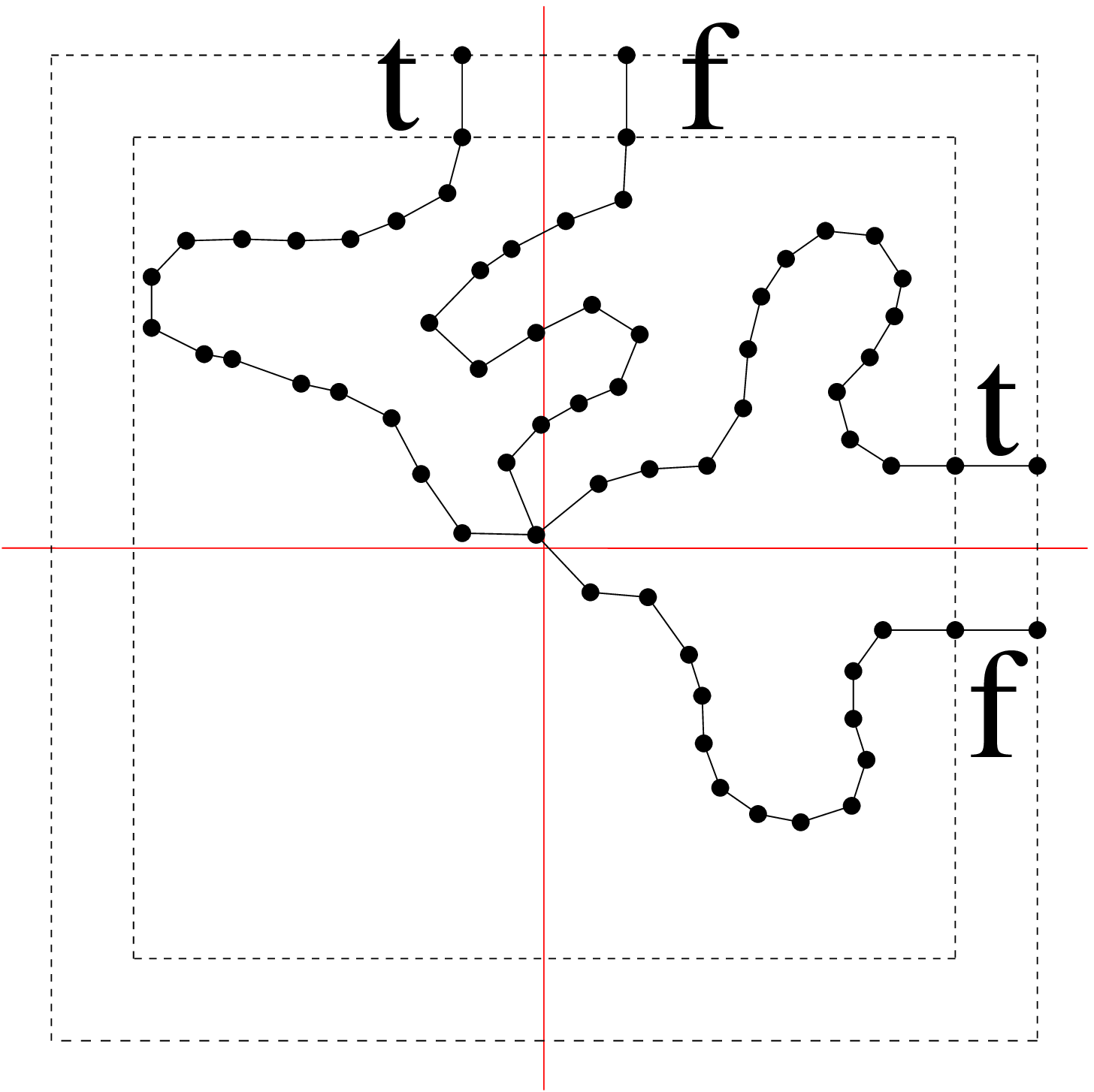} \hfill

\null\hfill
\includegraphics[width=100pt]{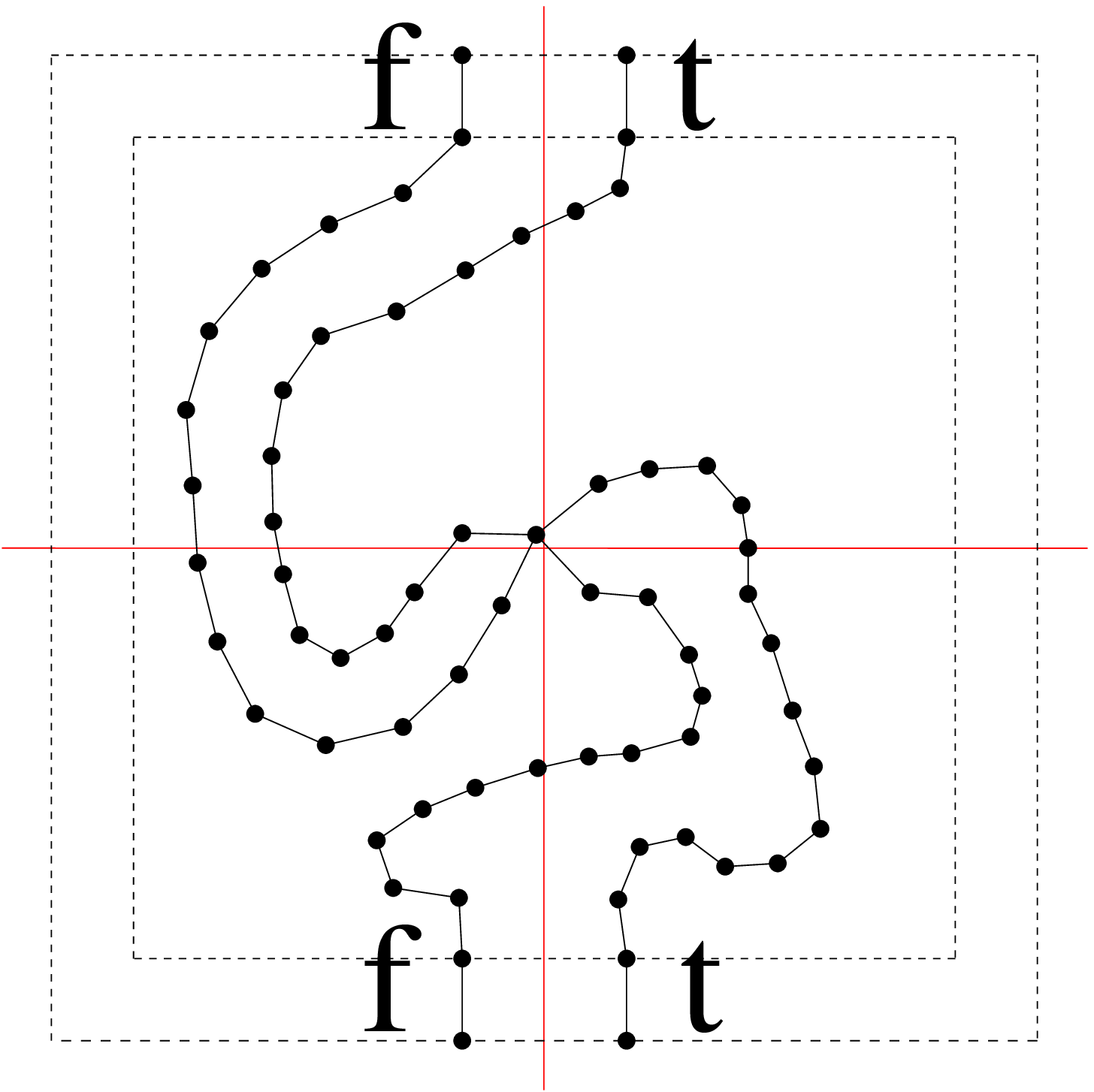} \hfill
\includegraphics[width=100pt]{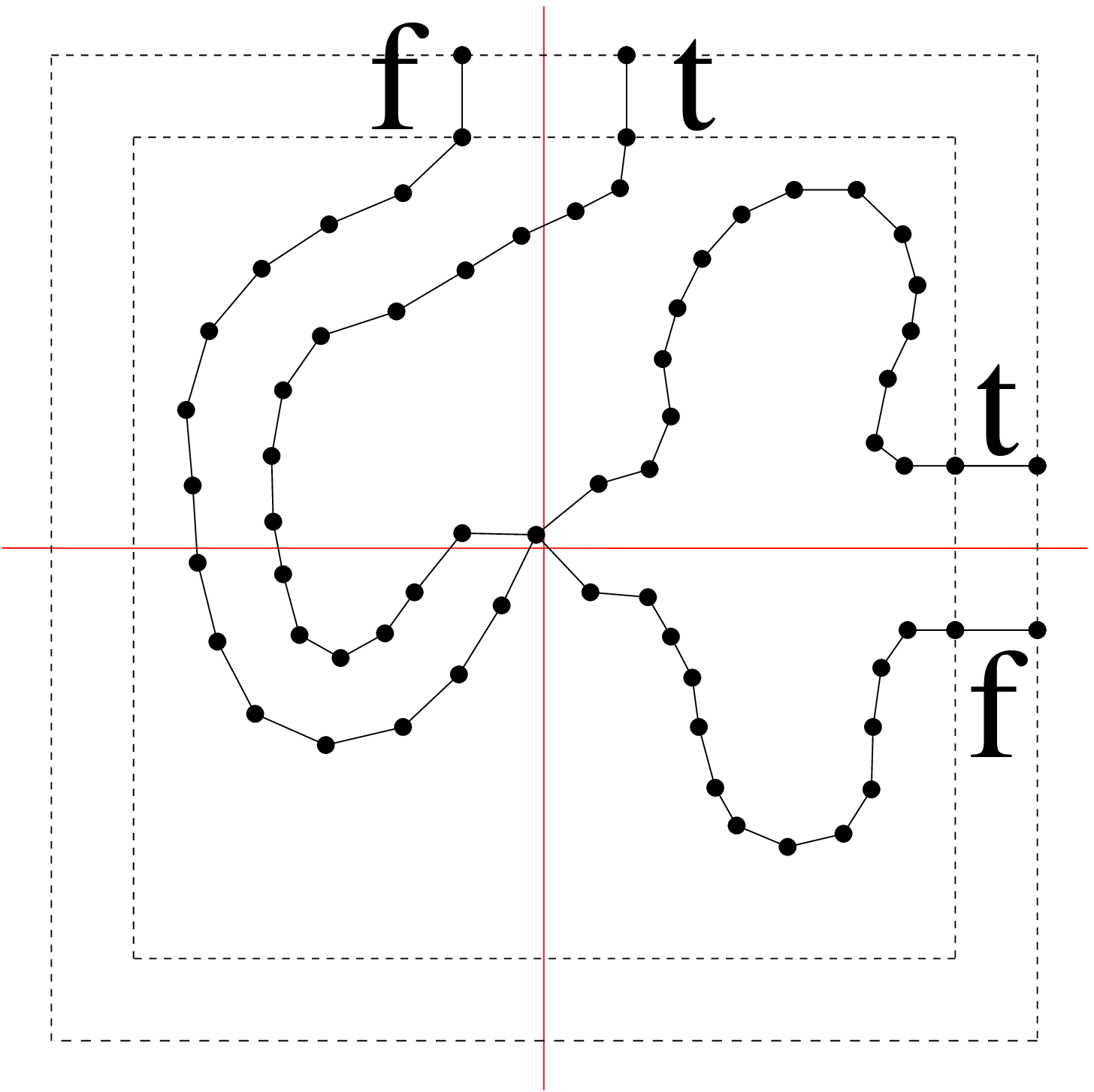} \hfill
\includegraphics[width=100pt]{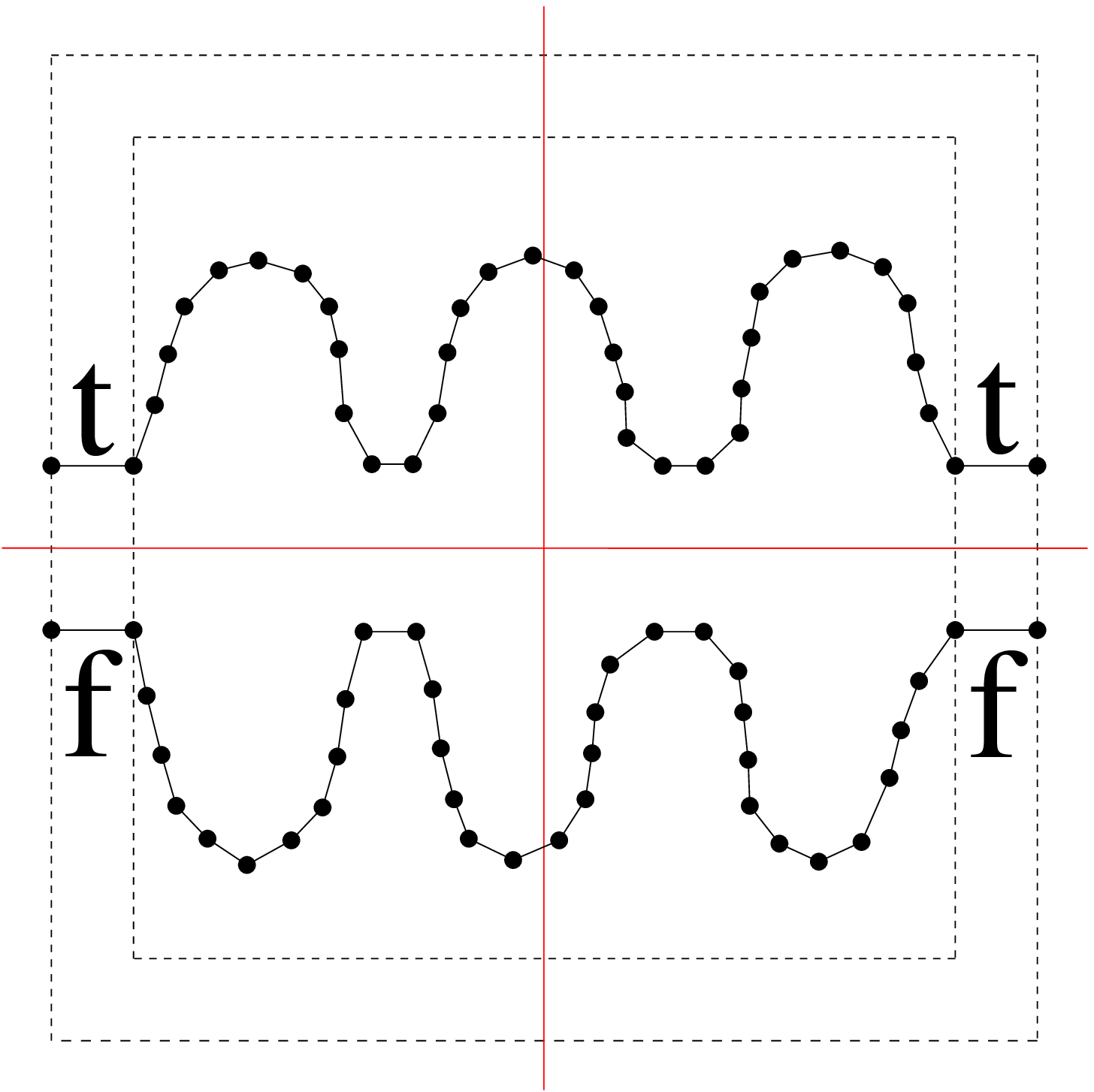}\hfill
\caption{Five tile embeddings for clause gadget $G^2_{c}$ and one tile embedding for the edge gadget $G_e$. Rotating and mirroring yields further tile embeddings.}
\label{clauseTiles2var}
\end{figure}

%\begin{figure}[hbt]
%\center
%  \begin{tabular}{c}
%      \includegraphics[height=4.6cm]{images/edge-tile.eps}
%  \end{tabular}
%\caption{A tile embedding for the edge gadget $G_e$. Rotating and mirroring yields further tile embeddings.}
%\label{edgePathTiles}
%\end{figure}

\begin{definition}\label{DefinitionTileTriple}
%A {\em tile triple} is a pair $(G,\rho_G)$, where $G=(V,E)$ is the edge gadget $G_{e}$ defined above and $\rho_G:V \to P_i$ for some $i \in \{1, \ldots, 6\}$ is a GUDG embedding for $G$, called {\em tile triple embedding}, such that $\{\{\rho_G(t_1),\rho_G(f_1)\},\{\rho_G(t_{37}),\rho_G(f_{37})\}\}$ is either $\{\{(-18,-1),(-18,1)\},\{(18,-1),(18,1)\}\}$, $\{\{(-1,-18),(1,-18)\},\{(-1,18),(1,18)\}\}$, $\{\{(-18,-1),(-18,1)\},\{(-1,-18),(1,-18)\}\}$, \newline $\{\{(-18,-1),(-18,1)\},\{(-1,18),(1,18)\}\}$, $\{\{(-1,-18),(1,-18)\},\{(18,-1),(18,1)\}\}$, or \newline $\{\{(-1,18),(1,18)\},\{(18,-1),(18,1)\}\}$, respectively.
A {\em tile triple} is a pair $(G_e,\rho)$, where $G_{e}=(V,E)$ is the edge gadget defined above and $\rho:V \to P_i$ for some $i \in \{1, \ldots, 6\}$ is a GUDG embedding for $G_e$, called {\em tile triple embedding}, such that $\{\{\rho(t_1),\rho(f_1)\},\{\rho(t_{37}),\rho(f_{37})\}\}$ is either $\{\{(-18,-1),(-18,1)\},\{(18,-1),(18,1)\}\}$,\newline 
$\{\{(-1,-18),(1,-18)\},\{(-1,18),(1,18)\}\}$,\newline
$\{\{(-18,-1),(-18,1)\},\{(-1,-18),(1,-18)\}\}$,\newline
$\{\{(-18,-1),(-18,1)\},\{(-1,18),(1,18)\}\}$,\newline
$\{\{(-1,-18),(1,-18)\},\{(18,-1),(18,1)\}\}$,\newline
or $\{\{(-1,18),(1,18)\},\{(18,-1),(18,1)\}\}$, respectively.

%\begin{center}
%\begin{tabular}{llll}
%1. $\{\{(-18,-1),(-18,1)\},\{(18,-1),(18,1)\}\}$, \\
%2. $\{\{(-1,-18),(1,-18)\},\{(-1,18),(1,18)\}\}$, \\
%3. $\{\{(-18,-1),(-18,1)\},\{(-1,-18),(1,-18)\}\}$, \\
%4. $\{\{(-18,-1),(-18,1)\},\{(-1,18),(1,18)\}\}$, \\
%5. $\{\{(-1,-18),(1,-18)\},\{(18,-1),(18,1)\}\}$, or \\
%6. $\{\{(-1,18),(1,18)\},\{(18,-1),(18,1)\}\}$, respectively.
%\end{tabular}
%\end{center}

Every non-$(t,f)$-vertex is mapped to a position of $P_i$ with a distance $> 1$ to all positions $\mathbb{R}^2 \setminus P_i$ outside of $P_i$.
\end{definition}

\begin{figure}[htb]
\null\hfill
\includegraphics[width=150pt]{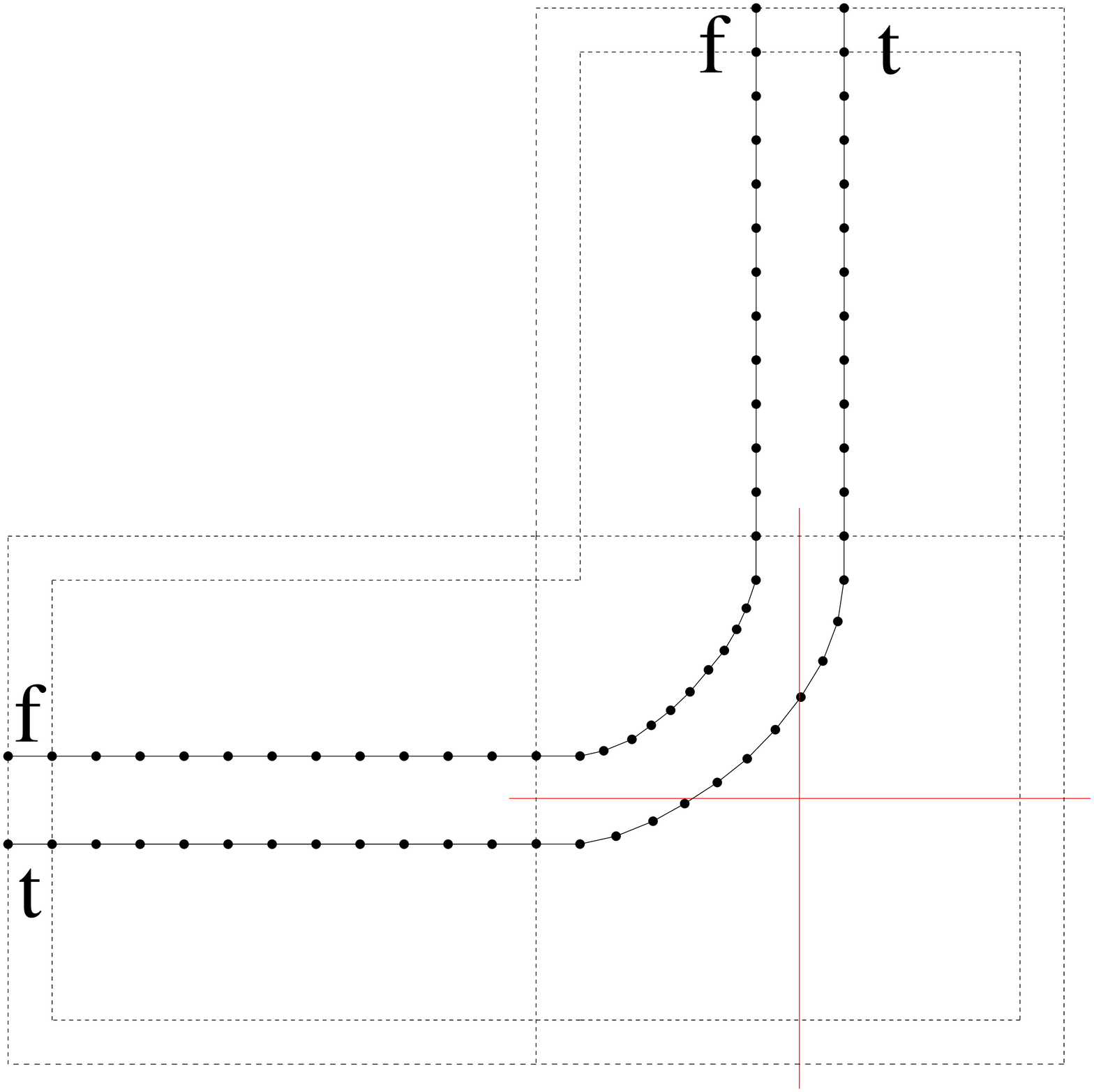} \hfill 
\includegraphics[width=150pt]{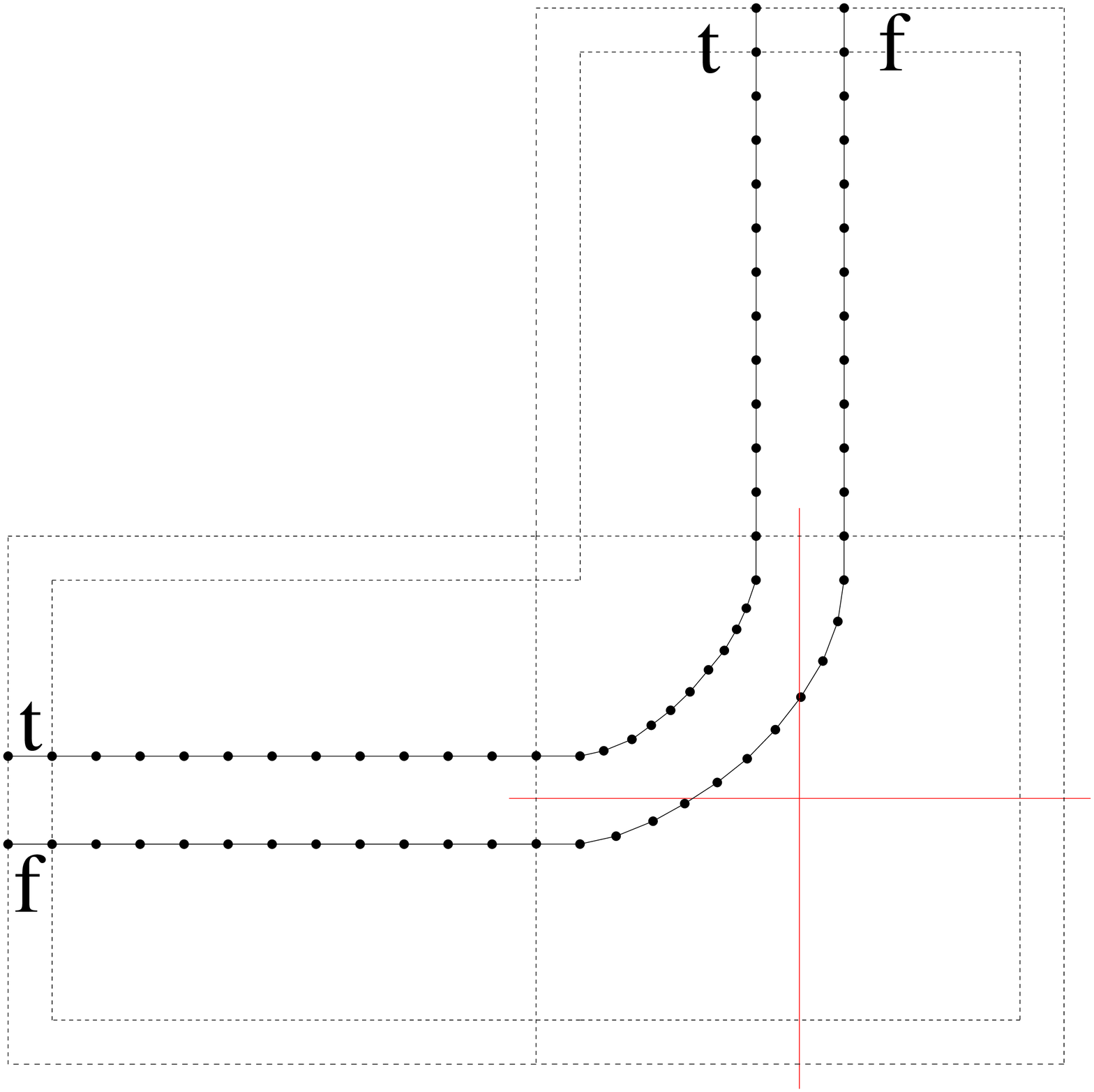} \hfill\null

\null\hfill
\includegraphics[width=250pt]{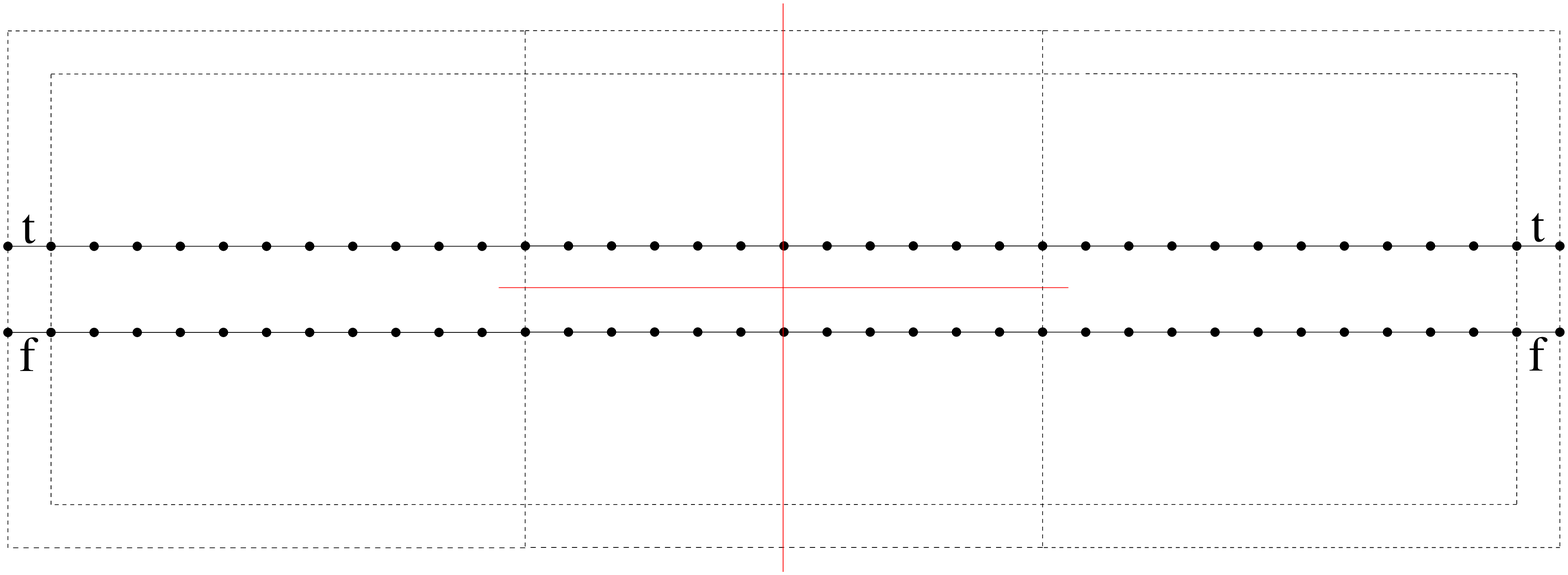} \hfill\null
\caption{Three tile triple embeddings for the edge gadget $G_e$, rotating and mirroring yields further tile triple embeddings for edge paths of length $2$.}
\label{tile-triples}
\medskip
\end{figure}

\begin{definition}
 The {\em orientation} of a $(t,f)$-vertex pair $(t,f)$ inside a tile (triple) $(G,\rho)$ is {\em TF}, if $\rho(t)$ is encounter before $\rho(f)$ on a clockwise traversal of the corresponding polygon and {\em FT} otherwise.
\end{definition}
%The orientation allows abstraction from the actual gadgets and tile embeddings by considering just the available connections.
\begin{definition}
 The {\em connection vector} of a tile $(G,\rho)$ is a $4$-tuple $(a_1,a_2,a_3,a_4)$ consisting of components $a_i = (b_i,c_i)$, where $b_i$ is the orientation of the $(t,f)$-vertex pair $(t_i,f_i)$ on the $i$-th side of the square \newline$[-6,6]\times [-6,6]$ or $\epsilon$, if none exists. The sides of the square are enumerated clockwise, starting from the top. If $G$ is a variable gadget and $b_i\not=\epsilon$, then $c_i\in\{\oplus,\ominus\}$ denotes whether $(t_i,f_i)$ is an $\oplus$-$(t,f)$-vertex pair or an $\ominus$-$(t,f)$-vertex pair. For all other cases set $c_i:=\emptyset$.
\end{definition}

Now define the equivalence relation $\sim_r$ on the space of connection vectors as follows: \newline$(a_1,a_2,a_3,a_4) \sim_r (b_1,b_2,b_3,b_4)$ if and only if there exists a $\delta\in\mathbb{N}$, such that for all $1\le i\le 4$, $a_i = b_{(i+\delta)~mod~4}$. The {\em type} of a tile is the equivalence class of its connection vector, denoted by one representative in square brackets. Table \ref{tileTypes} lists the types of all available tile embeddings.

\begin{table}\null\hfill
\def\arraystretch{1.2}
\begin{tabular}{|l|c|l|}\hline
 type & gadget & tile embedding	\\\hline
  $[(TF,\oplus),(\epsilon,\emptyset),(FT,\ominus),(FT,\oplus)]$	& $G_v^3$ &	fig.\ref{variableTiles3} ($1$)	\\\hline
  $[(TF,\oplus),(\epsilon,\emptyset),(TF,\ominus),(FT,\oplus)]$	& $G_v^3$ &	fig.\ref{variableTiles3} ($3$), mirrored\\\hline
  $[(\epsilon,\emptyset),(TF,\oplus),(FT,\ominus),(FT,\oplus)]$	& $G_v^3$ &	fig.\ref{variableTiles3} ($2$)	\\\hline
  $[(\epsilon,\emptyset),(TF,\oplus),(TF,\ominus),(FT,\oplus)]$	& $G_v^3$ &	fig.\ref{variableTiles3} ($2$), mirrored\\\hline
  $[(FT,\oplus),(TF,\oplus),(FT,\ominus),(\epsilon,\emptyset)]$	& $G_v^3$ &	fig.\ref{variableTiles3} ($3$)	\\\hline
  $[(FT,\oplus),(TF,\oplus),(TF,\ominus),(\epsilon,\emptyset)]$	& $G_v^3$ &	fig.\ref{variableTiles3} ($1$), mirrored\\\hline
  $[(\epsilon,\emptyset),(\epsilon,\emptyset),(FT,\ominus),(FT,\oplus)]$	& $G_{v,a}$ &	fig.\ref{variableTiles2ab} ($1$)	\\\hline
  $[(\epsilon,\emptyset),(\epsilon,\emptyset),(TF,\ominus),(FT,\oplus)]$	& $G_{v,b}$ &	fig.\ref{variableTiles2ab} ($6$), mirrored\\\hline
  $[(FT,\oplus),(\epsilon,\emptyset),(FT,\ominus),(\epsilon,\emptyset)]$	& $G_{v,a}$ &	fig.\ref{variableTiles2ab} ($2$)	\\\hline
  $[(FT,\oplus),(\epsilon,\emptyset),(TF,\ominus),(\epsilon,\emptyset)]$	& $G_{v,b}$ &	fig.\ref{variableTiles2ab} ($5$), mirrored\\\hline
  $[(\epsilon,\emptyset)(FT,\oplus),(FT,\ominus),(\epsilon,\emptyset)]$	& $G_{v,a}$ &	fig.\ref{variableTiles2ab} ($3$)	\\\hline
  $[(\epsilon,\emptyset)(FT,\oplus),(TF,\ominus),(\epsilon,\emptyset)]$	& $G_{v,b}$ &	fig.\ref{variableTiles2ab} ($4$), mirrored\\\hline

  $[(TF,\emptyset),(TF,\emptyset),(FT,\emptyset),(\epsilon,\emptyset)]$	& $G_c^3$ &	fig.\ref{clauseTiles3var} ($1$)	 \\\hline
  $[(TF,\emptyset),(FT,\emptyset),(FT,\emptyset),(\epsilon,\emptyset)]$	& $G_c^3$ &	fig.\ref{clauseTiles3var} ($1$), mirrored\\\hline
  $[(TF,\emptyset),(FT,\emptyset),(TF,\emptyset),(\epsilon,\emptyset)]$	& $G_c^3$ &	fig.\ref{clauseTiles3var} ($2$)	 \\\hline
  $[(FT,\emptyset),(TF,\emptyset),(FT,\emptyset),(\epsilon,\emptyset)]$	& $G_c^3$ &	fig.\ref{clauseTiles3var} ($2$), mirrored\\\hline
  $[(FT,\emptyset),(TF,\emptyset),(TF,\emptyset),(\epsilon,\emptyset)]$	& $G_c^3$ &	fig.\ref{clauseTiles3var} ($3$)	 \\\hline
  $[(FT,\emptyset),(FT,\emptyset),(TF,\emptyset),(\epsilon,\emptyset)]$	& $G_c^3$ &	fig.\ref{clauseTiles3var} ($3$), mirrored\\\hline
  $[(TF,\emptyset),(\epsilon,\emptyset),(TF,\emptyset),(\epsilon,\emptyset)]$	& $G_c^2$ &	fig.\ref{clauseTiles2var} ($1$)	 \\\hline
  $[(FT,\emptyset),(\epsilon,\emptyset),(FT,\emptyset),(\epsilon,\emptyset)]$	& $G_c^2$ &	fig.\ref{clauseTiles2var} ($1$), mirrored\\\hline
  $[(TF,\emptyset),(FT,\emptyset),(\epsilon,\emptyset),(\epsilon,\emptyset)]$	& $G_c^2$ &	fig.\ref{clauseTiles2var} ($2$)	 \\\hline
  $[(TF,\emptyset),(TF,\emptyset),(\epsilon,\emptyset),(\epsilon,\emptyset)]$	& $G_c^2$ &	fig.\ref{clauseTiles2var} ($3$)\\\hline
  $[(FT,\emptyset),(FT,\emptyset),(\epsilon,\emptyset),(\epsilon,\emptyset)]$	& $G_c^2$ &	fig.\ref{clauseTiles2var} ($3$), mirrored	\\\hline
  $[(FT,\emptyset),(\epsilon,\emptyset),(TF,\emptyset),(\epsilon,\emptyset)]$	& $G_c^2$ &	fig.\ref{clauseTiles2var} ($4$)\\\hline
  $[(TF,\emptyset),(\epsilon,\emptyset),(FT,\emptyset),(\epsilon,\emptyset)]$	& $G_c^2$ &	fig.\ref{clauseTiles2var} ($4$), mirrored	\\\hline
  $[(FT,\emptyset),(TF,\emptyset),(\epsilon,\emptyset),(\epsilon,\emptyset)]$	& $G_c^2$ &	fig.\ref{clauseTiles2var} ($5$)\\\hline
\end{tabular}\hfill\null\caption{The available tile types and the corresponding embeddings}
\label{tileTypes}
\end{table}

Next we place the tiles and tile triples according to the planar orthogonal grid drawing $(\rho_{G_{\psi'}},{\cal E})$ of the clause variable graph $G_{\psi'}$.

Let $(v,c)$ be an edge of $G_{\psi'}$ with edge path $\ep{v}{c} = ((x_1,y_1),\ldots,(x_n,y_n))$.
Edge $(v,c)$ is {\em from-left}, {\em from-right}, {\em from-top}, or {\em from-bottom} if $x_2 > x_1$, $x_2 < x_1$, $y_2 < y_1$, or $y_2 > y_1$, respectively. Analogously it is {\em to-left}, {\em to-right}, {\em to-top}, or {\em to-bottom} if $x_n < x_{n-1}$, $x_{n} > x_{n-1}$, $y_n > y_{n-1}$, or $y_{n} < y_{n-1}$. Edge $(v,c)$ is a $\oplus$-edge or a $\ominus$-edge if variable $v$ occurs in clause $c$ as a positive or negative literal.

A tile or tile triple $(G,\rho_{G})$ has a left-, right-, top-, or bottom-$(t,f)$-vertex pair if gadget $G$ has a $(t,f)$-vertex pair mapped to $((x_1,y_1),(x_2,y_2))$ such that $x_1,x_2 < 0$, $x_1,x_2 > 0$, $y_1,y_2 > 0$, or $y_1,y_2 < 0$, respectively.

Let $v$ be a variable vertex of $G_{\psi'}$ and $G$ be the copy of the gadget for $v$. Then a tile $(G,\rho_{G})$ for $G$ is chosen whose $(t,f)$-vertices are arranged according to the directions of the edges $e=(v,c)$ from $v$ in the following sense. If, for example, $e$ is a $\oplus$-from-left edge then tile $(G,\rho_{G})$ has a right-$\oplus$-$(t,f)$-vertex pair. If $e$ is a $\ominus$-from-top edge then tile $(G,\rho_{G})$ has a bottom-$\ominus$-$(t,f)$-vertex pair, etc.

Let $c$ be a clause vertex of $G_{\psi'}$ and $G$ be the copy of the gadget for $c$. Then a tile $(G,\rho_{G})$ for $G$ is chosen whose $(t,f)$-vertices are arranged according to the directions of the edges $e=(v,c)$ to $c$ in the following sense. If, for example, $e$ is a to-left edge then tile $(G,\rho_{G})$ has a right-$(t,f)$-vertex pair. If $e$ is a to-top edge then tile $(G,\rho_{G})$ has a bottom-$(t,f)$-vertex pair, etc.

Let $e$ be an edge of $G_{\psi'}$ and $G$ be the copy of the gadget for $e$. Then a tile or tile triple $(G,\rho_{G})$ for $G$ is chosen whose $(t,f)$-vertices are arranged according to the direction of edge $e$ in the following sense. If, for example, $e$ is a from-left-to-right edge then $(G,\rho_{G})$ has a left- and a right-$(t,f)$-vertex pair. If $e$ is a from-top-to-bottom edge then $(G,\rho_{G})$ has a top- and bottom-$(t,f)$-vertex pair, etc.

To construct a GUDG embedding $\rho_{H_{\psi}}$ for $H_{\psi}$ we place the vertices of the gadgets as follows.
\begin{enumerate}
\item
Let $G$ be a gadget inserted into $H_{\psi}$ for a variable or clause vertex $u$ of the clause variable graph $G_{\psi'}$, and let $(G,\rho_{G})$ be the selected tile. Then the GUDG embedding $\rho_{H_{\psi}}$ for the vertices $w$ from gadget $G$ is defined by
$\rho_{H_{\psi}}(w) := 24 \cdot \rho_{G_{\psi'}}(u) + \rho_{G}(w)$.

\item
Let $G$ be a gadget inserted into $H_{\psi}$ for an edge path \newline$\ep{u}{v}=((x_1,y_1),(x_2,y_2)) \in {\cal E}$ of length $1$, and let $(G,\rho_{G})$ be the selected tile for gadget $G$. Then the GUDG embedding $\rho_{H_{\psi}}$ for the vertices $w$ from gadget $G$ is defined by
$\rho_{H_{\psi}}(w) := 12 \cdot ((x_1,y_1)+(x_2,y_2)) + \rho_{G}(w)$.
%\[\rho_{H_{\psi}}(w) \quad := \quad 24 \cdot \frac{(x_1,y_1)+(x_2,y_2)}{2} + \rho_{G}(w).\]
\item
Let $G$ be a gadget inserted into $H_{\psi}$ for an edge path \newline$\ep{u}{v}=((x_1,y_1),(x_2,y_2),(x_3,y_3)) \in {\cal E}$ of length $2$, and let $(G,\rho_{G})$ be the selected tile triple for gadget $G$. Then the GUDG embedding $\rho_{H_{\psi}}$ for the vertices $w$ from gadget $G$ is defined by
$\rho_{H_{\psi}}(w) := 24 \cdot (x_2,y_2) + \rho_{G}(w)$.
%If $x_1 = x_2 = x_3$ or $y_1 = y_2 = y_3$, the tile has to be of Type 1, possibly rotated by $90$, $180$, or $270$ degree, otherwise it has to be of Type 2 or 3, possibly rotated by $90$, $180$, or $270$ degree. 
\end{enumerate}
Figure \ref{tileExample} shows an example of the placed tiles for the clause variable graph $G_{\psi'}$ of Figure \ref{edge-path-shortening}.

Graph $H_\psi$ is defined by the union of gadgets through identification of their $(t,f)$-vertices. Thus we have to show that it is always possible to select a tile (or tile triple) for every inserted copy of a gadget such that identified vertices are placed at the same position. 

Every tile and tile triple embedding is a GUDG embedding. A GUDG embedding for the inner part used in the variable gadgets is shown in Figure \ref{tileExample}. The distance between non-$(t,f)$-vertices from different tiles is always greater than $1$. Thus, the resulting embedding $\rho_{H_\psi}$ is a GUDG embedding for $H_\psi$. 

\begin{figure}[ht]
\null\hfill
\includegraphics[width=180pt]{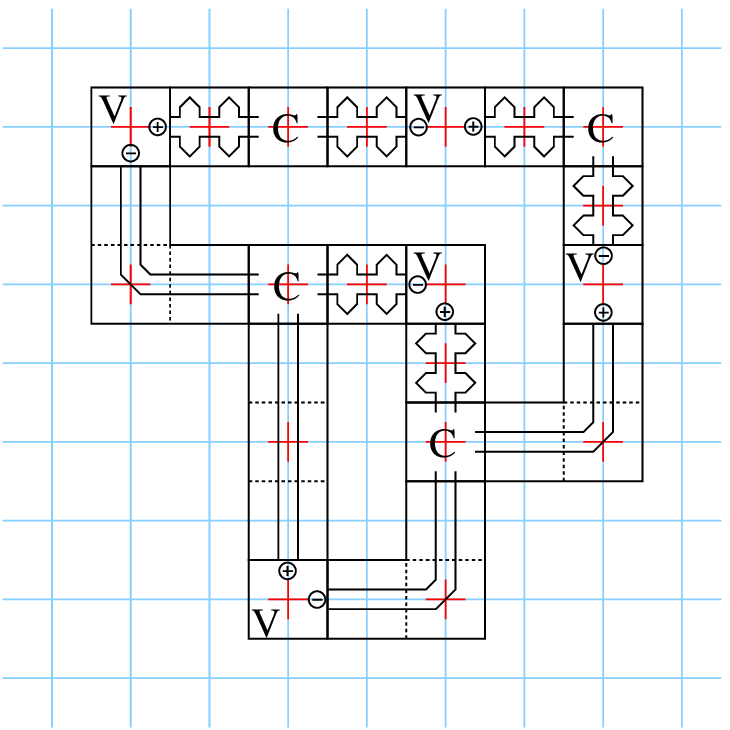} \hfill
\includegraphics[width=150pt]{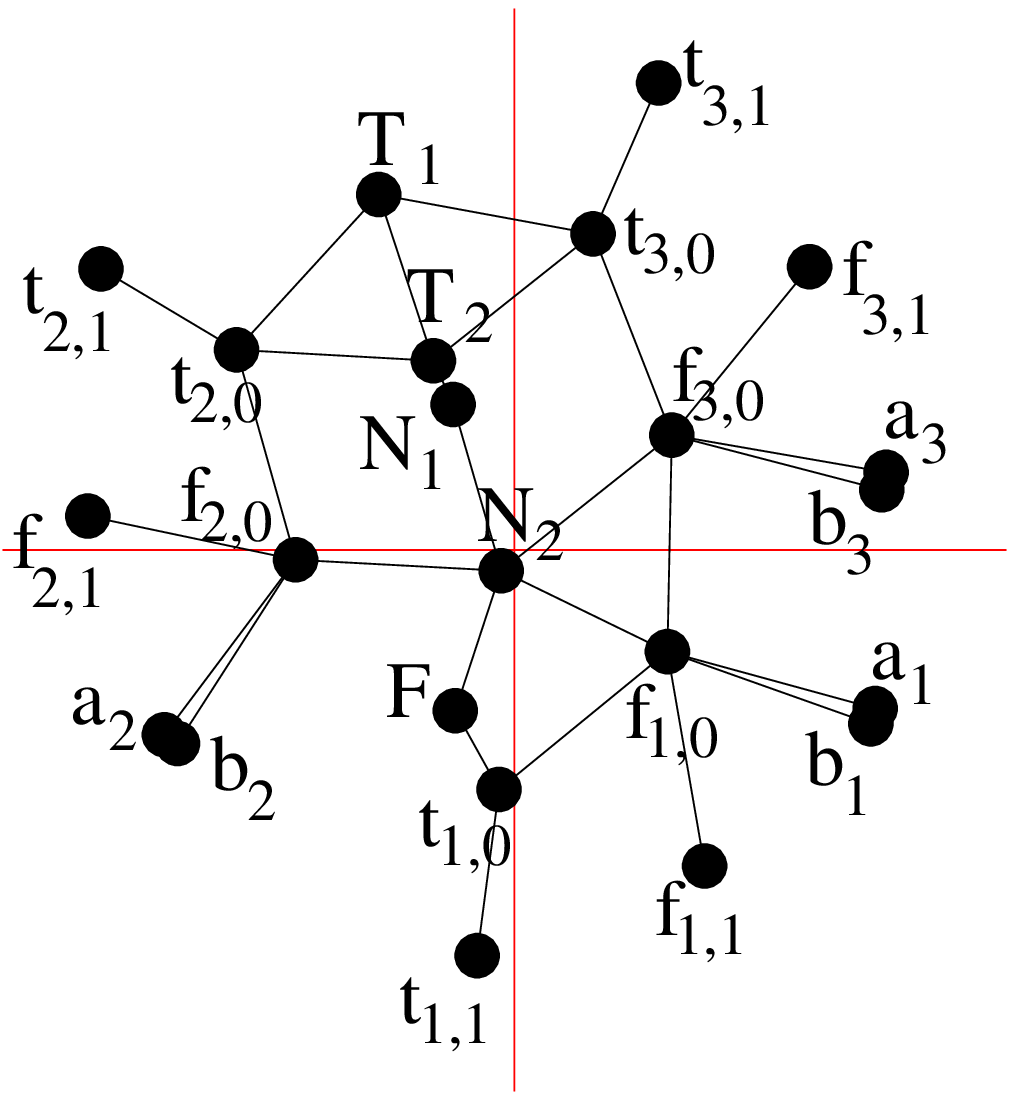} \hfill
\null
\caption{Left: A tile placement for the clause variable graph $G_{\psi'}$ of Figure \ref{edge-path-shortening}. Clause and variable gadgets are labeled $C$ and $V$, respectively. Right: A GUDG embedding for the inner part of the variable gadgets}%, see table \ref{coords} for the exact coordinates of the vertices.}
\label{tileExample}
\end{figure}
\begin{table}[htb]\centering
\begin{tabular}{|c|l|c||c|l|c|}
\hline
$v$ & $\rho(v)$ & $d$ & $v$ & $\rho(v)$ & $d$  \\\hline
$T_1$ & $(-0.7,1.61)$   & $(4,3,3)$ 	&  	$f_{1,0}$ & $(0.62,-0.48)$  & $(1,3,2)$	 \\\hline
$T_2$ & $(-0.45,0.85)$  & $(4,3,3)$ 	&  	$t_{2,0}$ & $(-1.35,0.9)$   & $(4,2,4)$	 \\\hline
$N_1$ & $(-0.36,0.65)$  & $(3,3,3)$ 	&  	$f_{2,0}$ & $(-1.08,-0.06)$ & $(3,1,3)$	 \\\hline
$T_1$ & $(-0.7,1.61)$   & $(4,3,3)$ 	&  	$t_{3,0}$ &	$(0.28,1.43)$   & $(3,4,2)$	 \\\hline
$F$   & $(-0.35,-0.75)$ & $(3,3,3)$ 	&   	$f_{3,0}$ & $(0.64,0.51)$   & $(2,3,1)$	 \\\hline
$a_1$ & $(1.57,-0.74)$  & $(0,4,3)$ 	&  	$t_{1,1}$ & $(-0.25,-1.87)$ & $(3,5,4)$	 \\\hline
$a_2$ & $(-1.68,-0.86)$ & $(4,0,4)$ 	&  	$f_{1,1}$ & $(0.79,-1.46)$  & $(2,4,3)$	 \\\hline
$a_3$ & $(1.62,0.34)$   & $(3,4,0)$	& 	$t_{2,1}$ &	$(-1.97,1.27)$  & $(5,3,5)$	 \\\hline
$b_1$     & $(1.55,-0.81)$  & $(1,4,3)$&	$f_{2,1}$ & $(-2.03,0.14)$  & $(4,2,4)$	\\\hline
$b_2$	 & $(-1.62,-0.9)$  & $(4,1,4)$ &	$t_{3,1}$ & $(0.58,2.12)$   & $(4,5,3)$	\\\hline
$b_3$	 & $(1.6,0.26)$    & $(3,4,1)$	&	$f_{3,1}$ & $(1.27,1.28)$   & $(3,4,2)$	\\\hline
$t_{1,0}$ & $(-0.15,-1.11)$ & $(2,4,3)$&		&		&			\\\hline
\end{tabular}\caption{Coordinates for the GUDG embedding in Figure \ref{tileExample} and the distances from $d = (d(v,a_1) , d(v,a_2) , d(v,a_3))$ from each vertex to $a_1$, $a_2$, and $a_3$}
\label{var-coords}
\end{table}

%$v$ & $\rho(v)$ & $d$ & $v$ & $\rho(v)$ & $d$ & $v$ & $\rho(v)$ & $d$ \\\hline
%$T_1$ & $(-0.7,1.61)$   & $(4,3,3)$ & $b_1$     & $(1.55,-0.81)$  & $(1,4,3)$ & $t_{3,0}$ &	$(0.28,1.43)$   & $(3,4,2)$ \\\hline
%$T_2$ & $(-0.45,0.85)$  & $(4,3,3)$ & $b_2$	    & $(-1.62,-0.9)$  & $(4,1,4)$ & $f_{3,0}$ & $(0.64,0.51)$   & $(2,3,1)$ \\\hline
%$N_1$ & $(-0.36,0.65)$  & $(3,3,3)$ & $b_3$	    & $(1.6,0.26)$    & $(3,4,1)$ & $t_{1,1}$ & $(-0.25,-1.87)$ & $(3,5,4)$ \\\hline
%$T_1$ & $(-0.7,1.61)$   & $(4,3,3)$ &           &                 &           & $f_{1,1}$ & $(0.79,-1.46)$  & $(2,4,3)$ \\\hline
%$F$   & $(-0.35,-0.75)$ & $(3,3,3)$ & $t_{1,0}$ & $(-0.15,-1.11)$ & $(2,4,3)$ & $t_{2,1}$ &	$(-1.97,1.27)$  & $(5,3,5)$ \\\hline
%$a_1$ & $(1.57,-0.74)$  & $(0,4,3)$ & $f_{1,0}$ & $(0.62,-0.48)$  & $(1,3,2)$ & $f_{2,1}$ & $(-2.03,0.14)$  & $(4,2,4)$ \\\hline
%$a_2$ & $(-1.68,-0.86)$ & $(4,0,4)$ & $t_{2,0}$ & $(-1.35,0.9)$   & $(4,2,4)$ & $t_{3,1}$ & $(0.58,2.12)$   & $(4,5,3)$	\\\hline
%$a_3$ & $(1.62,0.34)$   & $(3,4,0)$ & $f_{2,0}$ & $(-1.08,-0.06)$ & $(3,1,3)$ & $f_{3,1}$ & $(1.27,1.28)$   & $(3,4,2)$ \\\hline
\medskip

If $t$ is identified with $t'$ and $f$ is identified with $f'$ then there only are the following two cases left to be considered due to the restrictions for the placement of $(t,f)$-vertex pairs inside tiles and tile triples. Either $(\rho_{H_\psi}(t),\rho_{H_\psi}(f))=(\rho_{H_\psi}(t'),\rho_{H_\psi}(f'))$ or
$(\rho_{H_\psi}(t),\rho_{H_\psi}(f))=(\rho_{H_\psi}(f'),\rho_{H_\psi}(t'))$. In the first case there is no problem, the second case is called an {\em orientation conflict}. $\rho_{H_\psi}$ is a {\em valid} embedding if it has no orientation conflicts.

\begin{lemma}\label{variableTileProperty}
 Let $(G,\rho)$ be a variable tile. Then there also is a variable tile $(G',\rho')$ with:
\begin{enumerate}
  \item $(\rho(t_{1,14}),\rho(f_{1,14})) = (\rho'(f'_{1,14}),\rho'(t'_{1,14}))$
  \item	for all $\oplus$-$(t,f)$-vertex pairs $(t,f)$ in $(G,\rho)$ there is a $\oplus$-$(t,f)$-vertex pair $(t',f')$ in $(G',\rho')$ \newline such that $(\rho(t),\rho(f)) = (\rho'(t'),\rho'(f'))$.
\end{enumerate}
\end{lemma}

\begin{figure}[hbt]
\null\hfill
\includegraphics[height=1.5cm]{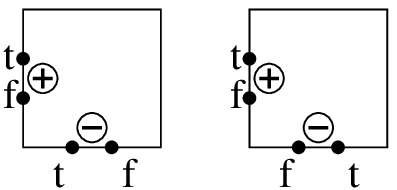} \hfill
\includegraphics[width=8.5cm]{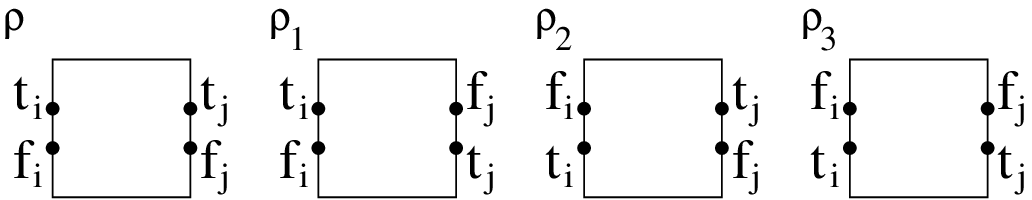} \hfill
\null 
\caption{Lemma \ref{variableTileProperty} and \ref{clauseTileProperty}: For every variable tile the orientation of the $\ominus$-$(t,f)$-vertex pair can be chosen freely by selecting another tile embedding and possibly another variable gadget. For every clause tile the orientation of two of its $(t,f)$-vertex pairs can be chosen freely by selection another tile embedding.}
\label{variableTilePropertyFig}
\end{figure}

%\begin{figure}[hbt]
%\center
%\includegraphics[width=10cm]{images/clauses/clauseTileProperty.eps}  
%\caption{Lemma \ref{clauseTileProperty}: For every clause tile the orientation of two of its $(t,f)$-vertex pairs can be chosen freely by selection another tile embedding.}
%\label{clauseTilePropertyFig}
%\end{figure}

\begin{lemma}\label{clauseTileProperty}
 Let $(G,\rho)$ be a clause tile and $(t_i,f_i)$, $(t_j,f_j)$ two distinct $(t,f)$-vertex pairs in $G$. Then there also are three other tiles $(G,\rho_1)$,$(G,\rho_2)$ and $(G,\rho_3)$ with:
\begin{description}
 \item $(\rho(t_i),\rho(f_i)) = (\rho_1(t_i),\rho_1(f_i))$, $(\rho(t_j),\rho(f_j)) = (\rho_1(f_j),\rho_1(t_j))$
 \item $(\rho(t_i),\rho(f_i)) = (\rho_2(f_i),\rho_2(t_i))$, $(\rho(t_j),\rho(f_j)) = (\rho_2(t_j),\rho_2(f_j))$
 \item $(\rho(t_i),\rho(f_i)) = (\rho_3(f_i),\rho_3(t_i))$, $(\rho(t_j),\rho(f_j)) = (\rho_3(f_j),\rho_3(t_j))$
\end{description}
\end{lemma}

Lemmata \ref{variableTileProperty} and \ref{clauseTileProperty} can easily be verified using table \ref{tileTypes}. Note that the tile in Lemma \ref{variableTileProperty} may contain a different variable gadget. This is the case for a variable that occurs only once as a positive literal.

\begin{theorem}
 $H_\psi$ is a GUDG if the correct gadgets for all variables are chosen during the assembly. The assembly can be computed in polynomial time.
\end{theorem}
\begin{proof}
The only ambiguity during the assembly of $H_\psi$ is the selection of either $G^2_{v,a}$ or $G^2_{v,b}$ for a variable that occurs once as a positive literal. To decide which gadget has to be used, we start with any selection of tiles for the variable and clause vertices that fit the directions of the edges of $G_\psi'$ as described above. By selecting tiles and tile triples for the edge gadgets that do not cause an orientation conflict with the adjacent variable tile, we can assume w.l.o.g. that all orientation conflicts are directly between clause and variable tiles. Now resolve all existing orientation conflicts as follows: By Lemma \ref{clauseTileProperty} we can select a clause tile $(G^2_c,\rho)$ for every clause containing two literals that does not cause any orientation conflicts with the adjacent two variable tiles. For every clause containing three literals we can select a clause tile $(G^3_c,\rho)$ that causes at most one orientation conflict. Since every clause with three literals contains at 
least one negative literal $\overline{x}$ we can restrict this orientation conflict to the variable tile $(G_x,\rho_x)$ for variable $x$. By Lemma \ref{variableTileProperty}, we can then select a variable tile $(G_x',\rho_x')$ that causes neither this orientation conflict nor one with another adjacent clause tile. $G_x'$ is the variable gadget that has to be used for variable $x$ during the assembly of $H_\psi$. Since every orientation conflict is solved by at most three local replacements of tiles and tile triples that never cause new orientation conflicts to appear the tile selection for a valid GUDG embedding, and by extend the GUDG $H_\psi$, can be computed in polynomial time.
\end{proof}

%%%%%%%%%%%%%%%%%%%%%%%%%%%%%%%%%%%%%%%%%%%%%%%%%%%%%%%%%%%%%%%%%%%%%%%%%%%%%%%%%%%%%%%%%
%%%%%%%%%%%%%%%%%%%%%%%%%%%%%%%%%%%%%%%%%%%%%%%%%%%%%%%%%%%%%%%%%%%%%%%%%%%%%%%%%%%%%%%%%
%%%%%%%%%%%%%%%%%%%%%%%%%%%%%%%%%%%%%%%%%%%%%%%%%%%%%%%%%%%%%%%%%%%%%%%%%%%%%%%%%%%%%%%%%
%%%%%%%%%%%%%%%%%%%%%%%%%%%%%%%%%%%%%%%%%%%%%%%%%%%%%%%%%%%%%%%%%%%%%%%%%%%%%%%%%%%%%%%%%
%%%%%%%%%%%%%%%%%%%%%%%%%%%%%%%%%%%%%%%%%%%%%%%%%%%%%%%%%%%%%%%%%%%%%%%%%%%%%%%%%%%%%%%%%

\subsection{Correctness}

In this section, we prove that there is a resolving set $S$ for $H_\psi$ of size at most $4\cdot |X'|$ if and only if there is a satisfying truth assignment for $\psi$.

\begin{definition}
Let $G=(V,E)$ be an undirected graph. A vertex $u \in V$ {\em resolves} a vertex pair $v,w \in V$ if $d(u,v) \not= d(u,w)$. A vertex pair $v,w$ is called {\em unsolved by set} $U \subseteq V$, if no vertex of $U$ resolves $v,w$.
\end{definition}

\begin{lemma}\label{min3landmarks}
Every resolving set $S$ for $H_\psi$ contains at least three vertices of each copy of a variable gadget $G$. $S$ contains vertex $a_i$ or $b_i$ for $i=1,\ldots,3$, see Figure \ref{variableGadgets}.
\end{lemma}

\begin{proof}
Suppose neither $a_i$ nor $b_i$ is in $S$ for some $i$. Since all shortest paths from any vertex of $S$ to $a_i$ and $b_i$ contain $f_{i,0}$ and $d(f_{i,0},a_i) = d(f_{i,0},b_i) = 1$, there is no vertex $s \in S$ that resolves the pair $a_i,b_i$.
\end{proof}

Let ${\cal F}$ be a set of vertices consisting of one vertex from $\{a_i,b_i\}$, $1 \leq i \leq 3$, for every copy of a variable gadget, i.e., $|{\cal F}| = 3\cdot |X'|$. The vertices in ${\cal F}$ are called {\em forced landmarks}.

%For a vertex $u$, let $V_G(u)$ be the vertices of the copy from gadget $G$ that contains $u$.

\begin{lemma}\label{resolveAll}
Every vertex pair $T_1,T_2$ and $N_1,F$ from the same copy of a variable gadget and every vertex pair $w_1,w_2$ from the same copy of a clause gadget are unsolved by ${\cal F}$. For all other vertex pairs $u,v$, there is a forced landmark of ${\cal F}$ that resolves $u,v$.
\end{lemma}

\begin{proof}
 Let $u,v$ be an arbitrary vertex pair of $H_\psi$.
  \begin{case} $u$ and $v$ belong to the same gadget copy $G$.
    \begin{subcase} $G$ is a variable gadget copy. \label{caseVarVar} Vertex pair $u,v = T_1,T_2$ is unsolved by ${\cal F}$, because every shortest path from a landmark of ${\cal F}$ to $T_1$ and $T_2$ passes one of the $f_{i,0}$ vertices and $d(T_1,f_{i,0})=d(T_2,f_{i,0})$ for $1\le i\le 3$. Analogously, vertex pair $u,v = N_1,F$ is unsolved by ${\cal F}$, since $d(N_1,f_{i,0})=d(F,f_{i,0})=2$ for $1\le i\le 3$. The vertex pairs $u,v = t_{i,j},f_{i,j+1}$, $1\le i\le 3$, $0\le j\le 13$ are resolved by the forced landmarks from the variable gadgets connected to this $(t,f)$-path pair via one clause gadget and two edge gadgets. All other vertex pairs $u,v$ are resolved by one of the forced landmarks of $G$, see table \ref{var-coords}.
    \end{subcase}
    \begin{subcase}$G$ is an edge gadget copy.\label{caseLineLine} All edge gadgets are basically extensions of the $(t,f)$-path pairs inside the clause and variable gadgets, so this case has already been covered in case \ref{caseVarVar}.
    \end{subcase}
    \begin{subcase}$G$ is a clause gadget copy.\label{caseClauseClause} The vertex pair $u,v = w_1,w_2$ is unsolved by ${\cal F}$, because all shortest paths from the forced landmarks to $w_1$ and $w_2$ enter $G$ via an $f$ path and pass vertex $m$. All other vertex pairs $u,v$ are resolved by at least one forced landmark in an adjacent variable gadget copy, see table \ref{clauseDistances}.
    \end{subcase}
  \end{case}
  \begin{case}$u$ and $v$ belong to different gadget copies.
    \begin{subcase}$u$ and $v$ belong to variable gadget copies $G_1$ and $G_2$, respectively. All forced landmarks of $G_1$ resolve the pair $u,v$: The distance from a forced landmark of $G_1$ to any vertex of $G_1$ is at most $18$ (e.g. $d(a_1,t_{2,14})$) while the distance to any vertex of $G_2$ is greater than $18$, because any shortest path has to traverse at least two edge gadgets and one clause gadget.
    \end{subcase}
    \begin{subcase}$u$ and $v$ both belong to edge gadget copies. Let $a$ and $b$ be the forced landmarks closest to $u$ and $v$, respectively. Also $d(a,u) < d(c,u)$ for all forced landmarks $c\not= a$, because all vertices at distance $d(a,u)$ from $u$ are either in a clause gadget, an edge gadget or the variable gadget that $a$ belongs to, hence none of them being another forced landmark. Now suppose that neither $a$ nor $b$ resolves the pair $u,v$, meaning $d(a,v) = d(a,u)$ and $d(b,u) = d(b,v)$. Together we get $d(a,v) = d(a,u) < d(b,u) = d(b,v)$, which contradicts to the assumption that $b$ is the landmark closest to $v$.
     \end{subcase}
    \begin{subcase}$u$ and $v$ belong to the clause gadget copies $G_c$ and $G_d$ for the clauses $c$ and $d$, respectively.
    \begin{subsubcase} There is a variable $x$ such that $c$ contains literal $x$ or $\overline{x}$ and $d$ does contain neither $x$ nor $\overline{x}$. The forced landmarks of variable gadget $G_x$ for $x$ resolve the pair $u,v$: All shortest paths from forced landmarks of $G_x$ to vertices of the adjacent clause gadget have a length of at most $86$ (e.g. $d(a_1,t_{1,15})$), while all shortest paths to vertices of a non-adjacent clause gadget have a length greater than $86$, since they have to traverse three edge gadgets, one clause gadget and one variable gadget.
    \end{subsubcase}
    \begin{subsubcase} All variables that occur as a literal in $c$ also occur as a literal in $d$ and vice versa.
     Since $c\not= d$ there is one variable $x$ such that $x\in c$ and $\overline{x}\in d$ or $\overline{x}\in c$ and $x\in d$. W.l.o.g. let $x\in c$, $\overline{x}\in d$. Furthermore let $a$ and $b$ be the forced landmarks in $G_x$ closest to $u$ and $v$, respectively. Then $d(u,a) < d(u,b)$ and $d(v,b) < d(v,a)$. Now assume $u,v$ is unsolved by ${\cal F}$. Then $d(v,a)=d(u,a)$ and $d(v,b)=d(u,b)$. Together we get $d(v,a) = d(u,a) < d(u,b) = d(v,b) < d(v,a)$ which is a contradiction. % Spezialfall: z.B. $c=\{x,y\}$ und $d=\{\overline{x},y\}$
    \end{subsubcase}
      \end{subcase}
    \begin{subcase}$u$ belongs to variable gadget copy $G_x$ and $v$ to edge gadget copy $G_e$.
      \begin{subsubcase}$G_x$ is adjacent to $G_e$. See case \ref{caseVarVar}.
      \end{subsubcase}
      \begin{subsubcase}\label{caseVarLineSub2}$G_x$ is not adjacent to $G_e$. Every forced landmarks $a$ in $G_x$ resolves $u,v$: The distance from $a$ to any vertex in $G_x$ is at most $18$ while the distance to any vertex in $G_e$ is greater than $18$, since the shortest path has to traverse at least one other edge gadget and one clause gadget.
      \end{subsubcase}
    \end{subcase}
    \begin{subcase}$u$ belongs to variable gadget copy $G_x$ and $v$ to clause gadget copy $G_c$. See case \ref{caseVarLineSub2}.
    \end{subcase}
    \begin{subcase}$u$ belongs to clause gadget copy $G_c$ and $v$ to edge gadget copy $G_e$.
      \begin{subsubcase}$G_c$ is adjacent to $G_e$. See case \ref{caseClauseClause}.
      \end{subsubcase}
      \begin{subsubcase}$G_c$ is not adjacent to $G_e$. If $v$ belongs to a $(t,f)$-vertex pair of $G_e$ then this case has already been covered in one of the previous cases, because $v$ has been identified with a vertex from a clause or variable gadget. So we can assume that $v\in \{t_2,\cdots ,t_{36},f_2,\cdots ,f_{36}\}$. Let $a$ be the forced landmark closest to $v$ inside a variable gadget adjacent to $G_e$. Then $a$ resolves $u,v$, because $d(a,v) \le 51$ while the distance from $a$ to any vertex in $G_c$ is at least $52$.% (e.g. $d(a_3,f_{1,14}) + d(f_1,f_{37})$).
      \end{subsubcase}
    \end{subcase}
  \end{case}
\end{proof}
\begin{table}[htb]\null\hfill
 \begin{tabular}{|c|c||c|c||c|c||c|c|}\hline
$v$	&	$(d_x,d_y,d_z)$		 & 
$v$	&	$(d_x,d_y,d_z)$	 	 & 
$v$	&	$(d_x,d_y,d_z)$		 & 
$v$	&	$(d_x,d_y,d_z)$		 \\\hline\hline 
$m$	&	$(66,66,66)$		 & 
$w_1$	&	$(67,67,67)$	 	 & 
$w_2$	&	$(67,67,67)$		 & 
$c_1$	&	$(67,68,68)$		 \\\hline
$c_2$	&	$(68,68,67)$		 &
$c_3$	&	$(68,67,68)$		 & 
$t_{1,1}$	&	$(66,69,69)$	 &
$t_{1,2}$	&	$(65,70,70)$	 \\\hline
$t_{1,3}$	&	$(64,71,71)$	&
$t_{1,4}$	&	$(63,72,72)$	&
$t_{1,5}$	&	$(62,73,73)$	 & 
$t_{1,6}$	&	$(61,74,74)$	 \\\hline
$t_{1,7}$	&	$(60,75,75)$	&
$t_{1,8}$	&	$(59,76,76)$	 & 
$t_{1,9}$	&	$(58,77,77)$	&
$t_{1,10}$	&	$(57,78,78)$	 \\\hline
$t_{1,11}$	&	$(56,79,79)$	&
$t_{1,12}$	&	$(55,80,80)$	 & 
$t_{1,13}$	&	$(54,81,81)$	 & 
$t_{1,14}$	&	$(53,82,82)$	\\\hline
$t_{1,15}$	&	$(52,82,82)$	 & 
$f_{1,1}$	&	$(65,67,67)$	 & 
$f_{1,2}$	&	$(64,68,68)$	 & 
$f_{1,3}$	&	$(63,69,69)$	\\\hline
$f_{1,4}$	&	$(62,70,70)$	 &
$f_{1,5}$	&	$(61,71,71)$	 & 
$f_{1,6}$	&	$(60,72,72)$	 & 
$f_{1,7}$	&	$(59,73,73)$	\\\hline
$f_{1,8}$	&	$(58,74,74)$	 & 
$f_{1,9}$	&	$(57,75,75)$	&
$f_{1,10}$	&	$(56,76,76)$	 & 
$f_{1,11}$	&	$(55,77,77)$	\\\hline
$f_{1,12}$	&	$(54,78,78)$	 & 
$f_{1,13}$	&	$(53,79,79)$	 & 
$f_{1,14}$	&	$(52,80,80)$	&
$f_{1,15}$	&	$(51,81,81)$	 \\\hline
$t_{2,1}$	&	$(69,66,69)$	 & 
$t_{2,2}$	&	$(70,65,70)$	& 
$t_{2,3}$	&	$(71,64,71)$	 & 
$t_{2,4}$	&	$(72,63,72)$	\\\hline
$t_{2,5}$	&	$(73,62,73)$	 & 
$t_{2,6}$	&	$(74,61,74)$	& 
$t_{2,7}$	&	$(75,60,75)$	 & 
$t_{2,8}$	&	$(76,59,76)$	\\\hline
$t_{2,9}$	&	$(77,58,77)$	 &
$t_{2,10}$	&	$(78,57,78)$	 & 
$t_{2,11}$	&	$(79,56,79)$	 & 
$t_{2,12}$	&	$(80,55,80)$	\\\hline
$t_{2,13}$	&	$(81,54,81)$	 & 
$t_{2,14}$	&	$(82,53,82)$	 &
$t_{2,15}$	&	$(82,52,82)$	 & 
$f_{2,1}$	&	$(67,65,67)$	\\\hline
$f_{2,2}$	&	$(68,64,68)$	&
$f_{2,3}$	&	$(69,63,69)$	 & 
$f_{2,4}$	&	$(70,62,70)$	 &
$f_{2,5}$	&	$(71,61,71)$	 \\\hline
$f_{2,6}$	&	$(72,60,72)$	&
$f_{2,7}$	&	$(73,59,73)$	 & 
$f_{2,8}$	&	$(74,58,74)$	 & 
$f_{2,9}$	&	$(75,57,75)$	 \\\hline
$f_{2,10}$	&	$(76,56,76)$	&
$f_{2,11}$	&	$(77,55,77)$	 & 
$f_{2,12}$	&	$(78,54,78)$	 & 
$f_{2,13}$	&	$(79,53,79)$	 \\\hline
$f_{2,14}$	&	$(80,52,80)$	 &
$f_{2,15}$	&	$(81,51,81)$	 &
$t_{3,1}$	&	$(69,69,66)$	&
$t_{3,2}$	&	$(70,70,65)$	 \\\hline
$t_{3,3}$	&	$(71,71,64)$	 & 
$t_{3,4}$	&	$(72,72,63)$	 &
$t_{3,5}$	&	$(73,73,62)$	&
$t_{3,6}$	&	$(74,74,61)$	\\\hline
$t_{3,7}$	&	$(75,75,60)$	 & 
$t_{3,8}$	&	$(76,76,59)$	 & 
$t_{3,9}$	&	$(77,77,58)$	&
$t_{3,10}$	&	$(78,78,57)$	\\\hline
$t_{3,11}$	&	$(79,79,56)$	 & 
$t_{3,12}$	&	$(80,80,55)$	 & 
$t_{3,13}$	&	$(81,81,54)$	&
$t_{3,14}$	&	$(82,82,53)$	\\\hline
$t_{3,15}$	&	$(82,82,52)$	 & 
$f_{3,1}$	&	$(67,67,65)$	 & 
$f_{3,2}$	&	$(68,68,64)$	&
$f_{3,3}$	&	$(69,69,63)$	\\\hline
$f_{3,4}$	&	$(70,70,62)$	&
$f_{3,5}$	&	$(71,71,61)$	 & 
$f_{3,6}$	&	$(72,72,60)$	 & 
$f_{3,7}$	&	$(73,73,59)$	\\\hline
$f_{3,8}$	&	$(74,74,58)$	 & 
$f_{3,9}$	&	$(75,75,57)$	&
$f_{3,10}$	&	$(76,76,56)$	 & 
$f_{3,11}$	&	$(77,77,55)$	 \\\hline
$f_{3,12}$	&	$(78,78,54)$	 & 
$f_{3,13}$	&	$(79,79,53)$	 & 
$f_{3,14}$	&	$(80,80,52)$	 &
$f_{3,15}$	&	$(81,81,51)$	 \\\hline  
 \end{tabular}\hfill\null
\caption{Hopdistances $d_x,d_y,d_z$ between vertices of a copy of $G^3_c$ and their closest forced landmark inside adjacent variable gadget copies for variables $x,y,z$}
\label{clauseDistances}
\end{table}

\begin{lemma}\label{resolveT1T2}
%The vertex pair $T_1,T_2$ can only be resolved by choosing $T_1$, $T_2$, $N_1$, $N_2$ or $F$ of the same gadget copy $G_x$ as a landmark.
The vertex pair $T_1,T_2$ can only be resolved by $T_1$, $T_2$, $N_1$, $N_2$ or $F$ of the same gadget copy $G_x$.
\end{lemma}
\begin{proof}
Anyone of these vertices obviously resolves $T_1,T_2$ and except for $N_2$ they also resolve $N_1,F$. Furthermore there are shortest paths from $T_1$ and $T_2$ to any other vertex $v$ in $G_x$ that cross either $t_{2,0}$ or $t_{3,0}$, meaning that $v$ does not resolve $T_1,T_2$, because $d(T_1,v) = min\{d(t_{2,0},v) , d(t_{3,0},v)\} + 1 = d(T_2,v)$.
\end{proof}

%Choosing $T_1$ or $T_2$ as a landmark is equivalent to setting the variable to $true$ and choosing $F$ to setting it to $false$.

\begin{corollary}
The Lemmata \ref{min3landmarks}, \ref{resolveAll} and \ref{resolveT1T2} imply that every resolving set for $H_\psi$ contains at least $4\cdot |X'|$ vertices.
\end{corollary}

\begin{theorem}
{\sc GUDG Metric Dimension} is NP-complete.
\end{theorem}
\begin{proof}
Let $\psi = (X, {\cal C})$ be an instance for {\sc $1$-Negative Planar 3-Sat} and $(H_\psi,4\cdot |X'|)$, $H_\psi=(V_H,E_H)$ the instance for {\sc GUDG Metric Dimension} as described above.
Then $\psi$ is satisfiable    $\Leftrightarrow    \exists$ resolving set $S\subseteq V_H : |S|\le 4\cdot |X'|$.

  $\Rightarrow$:
    Because $\psi$ is satisfiable, the altered instance $\psi'$ with variable set $X'=\{x_1, \dots ,x_n\}$ and clause set ${\cal C'}$ is also satisfiable. Let $A: X\to \{true,false\}$ be a satisfying truth assignment for $\psi'$.
    For every variable $x\in X'$ choose $T_1$ of the corresponding variable gadget copy as a landmark if $A(x)=true$ and $F$ otherwise.
    Together with ${\cal F}$ this set of vertices is a resolving set:
    According to Lemmata \ref{resolveAll} and \ref{resolveT1T2} all pairs are resolved except for $w_1,w_2$ of each clause gadget copy.
    Since $A$ satisfies $\psi'$ each clause gadget copy is connected to either the $\ominus$-$(t,f)$-path pair of a variable gadget copy in which $F$ has been chosen as a landmark or to one of the $\oplus$-$(t,f)$-path pairs of a variable gadget copy in which $T_1$ has been chosen as a landmark.
    In both cases, for some $i\in\{0,1\}$, the shortest path from this landmark to $w_{1+i}$ enters the clause gadget through the $t$ path while the shortest path to $w_{2-i}$ enters the gadget through the $f$ path, thus resolving $w_1,w_2$.
%In both cases the shortest path from this landmark to one of the $w_i$ vertices enters the clause gadget through the $t$ path while the shortest path to the other $w_j$ vertex enters the gadget through the $f$ path, thus resolving $\{w_1,w_2\}$.

  $\Leftarrow$:
    Let $S$ be a resolving set for $H_\psi$ with $4\cdot |X'|$ vertices. According to Lemmata \ref{min3landmarks} and \ref{resolveT1T2} $S$ contains $3\cdot |X'|$ forced landmarks and one of the vertices $T_1$, $T_2$, $N_1$, $N_2$ or $F$ for every variable in $X'$. Now construct a truth assignment $A$ for $\psi'$ by setting $A(x)$ to $false$ if the fourth landmark in $x$'s variable gadget copy is $F$ and $true$ otherwise. Then $A$ satisfies $\psi'$:
    $S$ resolves $w_1,w_2$ in every clause gadget copy, so there is one landmark $s\in S$ for each copy such that $d(s,w_1)\not= d(s,w_2)$. This implies that the shortest path $p$ from $s$ to either $w_1$ or $w_2$ does not contain $m$ and enters the clause gadget through the adjacent $t$ path. So $s$ has to be inside the variable gadget that is connected to the clause gadget through one edge gadget:

    Assume that it is not. Then $p$ has to traverse an entire clause gadget $G_c$ and another variable gadget $G_v$ first. Independent of where $p$ enters $G_c$, on a shortest path to a vertex outside of $G_v$ it is never shorter to leave $G_c$ through a $t$ path than crossing $G_c$'s $m$ node and leaving through the $f$ path. So $p$ enters $G_v$ through an $f$ path, maybe contains the $N_2$ node of $G_v$ and leaves through another $f$ path, hence not resolving $w_1,w_2$.

    %Since the shortest path from $s$ to one of the $w_i$ nodes enters the clause gadget through the $t$ path and $s$ also resolves $\{T_1,T_2\}$ in its variable gadget, $A(x)$ has been set to $false$ if the clause gadget is connected to the $\ominus$ $(t,f)$-path pair of the variable gadget and $true$ otherwise.
    Since $p$ enters the clause gadget through the $t$ path and $s$ also resolves $T_1,T_2$ in its variable gadget, $A(x)$ has been set to $false$ if the clause gadget is connected to the $\ominus$-$(t,f)$-path pair of the variable gadget and $true$ otherwise.
    So every clause contains at least one positive literal whose variable is set to $true$ or one negative literal whose variable is set to $false$, meaning that $\psi'$ is satisfied. And because $\psi'$ is satisfiable, the original instance $\psi$ is satisfiable.
\end{proof}

\clearpage
\bibliographystyle{plain}
\bibliography{literature}

\end{document}